\documentclass[11pt,twoside]{article}
\usepackage[english]{babel}
\usepackage{multirow}
\usepackage{polski}
\usepackage{mathtools}
\usepackage{amsfonts}
\usepackage{booktabs}
\usepackage{amsmath}
\usepackage{amssymb}
\usepackage{comment}
\usepackage{makeidx}
\usepackage{amsthm}
\usepackage{color}
\usepackage{array}
\usepackage{xy}
\usepackage{xypic}
\input{xy}
\usepackage{tikz}
\usetikzlibrary{calc,decorations.markings,positioning}

\xyoption{arc}
\xyoption{all}
\definecolor{orange}{rgb}{0.698,0.133,0.133} 
\definecolor{green}{rgb}{0.33,0.42,0.18} 
\definecolor{greenf}{rgb}{0.13,0.55,0.13} 
\definecolor{newcolor}{rgb}{0.556364, 0.367273, 0.0763636}
\usepackage{varioref}
\usepackage{nameref}
 \usepackage[bookmarks=true,colorlinks=true]{hyperref}
\hypersetup{colorlinks, breaklinks,urlcolor=blue, citecolor=newcolor,linkcolor=newcolor}

\makeatletter
\newcommand{\bdots}{\mathinner{\mkern1mu\raise\p@\vbox{\kern8\p@\hbox{.}}\mkern2mu\raise4\p@\hbox{.}\mkern2mu\raise7\p@\hbox{.}\mkern1mu}}
\makeatother

\newtheorem{theorem}{Theorem}[section]
\newtheorem{definition}[theorem]{Definition}

\newtheorem{proposition}[theorem]{Proposition}

\newtheorem{corollary}[theorem]{Corollary}
\newtheorem{lemma}[theorem]{Lemma}
\newtheorem{ur:remark}[theorem]{Remark}
\newenvironment{remark}{\begin{ur:remark}\rm}{\end{ur:remark}}
\newtheorem{notation}[theorem]{Notation}
\newtheorem{ur:remarks}[theorem]{Remarks}
\newenvironment{remarks}{\begin{ur:remarks}\rm}{\end{ur:remarks}}
\newtheorem{ur:exemple}[theorem]{Example}
\newenvironment{example}{\begin{ur:exemple}\rm}{\end{ur:exemple}}
\newtheorem{ur:examples}[theorem]{Examples}
\newenvironment{examples}{\begin{ur:examples}\rm}{\end{ur:examples}}

\newdir^{ (}{{}*!/-5pt/@^{(}}

\newtheorem{pr}{}

\newcommand{\bpr}{\begin{pr} \begin{rm}}
\newcommand{\epr}{\end{rm} \end{pr}}

\newcommand{\nc}{{\mathbb C}}

\newcommand{\fraks}{{\mathfrak S}}
\newcommand{\nz}{{\mathbb Z}}

\newcommand{\nn}{{\mathbb N}}

\def\ri{{\mathrm i}}

\newcommand{\p}[1]{\left({#1}\right)}
\newcommand{\stirlingtwo}[2]{\genfrac{\{}{\}}{0pt}{}{#1}{#2}}
\newcommand{\av}[1]{\left|{#1}\right|}

\newcommand{\qu}[1]{\left[{#1}\right]}
\newcommand{\bm}[1]{\mbox{\boldmath${#1}$\unboldmath}}
\newcommand{\mb}[1]{\mathbf{{#1}}}
\newcommand{\brr}[1]{\left\{{#1}\right\}}

\def\gl{{\mathrm{GL}}}

\def\Mat{{\mathrm{Mat}}}
\def\Id{{\mathrm{Id}}}
\def\Hom{{\mathrm{Hom}}}

\def\Sym{{\mathrm{Sym}}}
\def\sym{{\mathrm{sym}}}

\def\expodot{\exp_{\odot}}
\def\odots{\odot\cdots\odot}

\newlength{\Oldarrayrulewidth}
\newcommand{\Cline}[2]{%
  \noalign{\global\setlength{\Oldarrayrulewidth}{\arrayrulewidth}}%
  \noalign{\global\setlength{\arrayrulewidth}{#1}}\cline{#2}%
  \noalign{\global\setlength{\arrayrulewidth}{\Oldarrayrulewidth}}}

 \usepackage{booktabs}

    {\end{list}}%
\providecommand{\bysame}{\leavevmode\hbox to3em{\hrulefill}\thinspace}

\begin{document}
\title{{\huge{Linearised Higher Variational Equations}}}
\author{{{\Large{Sergi Simon}}} \\
Department of Mathematics\\
University of Portsmouth \\
Lion Gate Bldg, Lion Terrace\\
Portsmouth PO1 3HF, UK}

\maketitle

\begin{abstract}
This work explores the tensor and combinatorial constructs underlying the linearised higher-order variational equations $\mathrm{LVE}_{\psi}^k$ of
a generic autonomous system along a particular solution $\bm\psi$. The main result of this paper is a compact yet explicit and computationally amenable form for said variational systems and their monodromy matrices. Alternatively, the same methods are useful to retrieve, and sometimes simplify, systems satisfied by the coefficients of the Taylor expansion of a formal first integral for a given dynamical system. This is done in preparation for further results within Ziglin-Morales-Ramis theory, specifically those of a constructive nature.
\end{abstract}


\section{Motivation and first definitions}
\subsection{Introduction}
\emph{Integrability}, an informal word for reasonably simple solvability, is an important problem in Dynamical Systems. Its opposite phenomenon, and specifically low predictability with respect to time, is usually summarised under the term \emph{chaos}. If the system is Hamiltonian, as are most problems in Mechanics, the ``chaos vs solvability" disjunctive is doubly advantageous. On one hand, it is amenable to the techniques of Symplectic Geometry. On the other, theory and empirics yield the specific, thus observable integrability condition described in \S \ref{moralesramisziglin}.

The introduction of the algebraic approach by Ziglin, Morales-Ruiz and Ramis produced hallmark contributions to the study of the integrability of Hamiltonian systems \cite{Au01a,Mo99a,MoralesRamis,Ziglin}, essentially couched on a study of the invariants of a given matrix group, associated to a linear system: the \emph{first-order variational equations} introduced in \ref{intro1}.
A second step forward was carried out by Morales-Ruiz, Ramis and Sim\'o (\cite{MoRaSi07a}) in order to extend the preceding Galoisian framework to the groups of the higher-order variational equations along a particular solution.

The second step described above is the driving force behind this paper. A constructive version of the Morales-Ramis-Sim\'o theorem was already started in \cite{Ap10a} and tangentially tackled from another viewpoint in \cite{ABSW} (see \S \ref{firstintegrals}) and the present work aims at expanding this effort by offering a closed-form expression for the linearised higher variationals. May the reader bear in mind that \emph{nowhere from \S \ref{symsection} onwards, except for \S \ref{samexample}, is the system required to be Hamiltonian}.

\subsection{Dynamical systems and variational equations}\label{intro1}

In accordance with results described in \S \ref{moralesramisziglin} and thereafter, we need to observe the following convention outside of Sections \ref{symsection} and \ref{syminfsection}: dependent and independent variables for all dynamical systems will be allowed to be \emph{complex}. Any open set $T\subseteq \mathbb{P}^1_\nc$ is an admissible domain for the time variable, embedded into the Riemann sphere to include $t=\infty$ as a valid singularity.
Consider an autonomous holomorphic dynamical system:
\begin{equation} \label{DS}\tag{\textrm{DS}}
\dot{\bm{z}} = X\left( \bm{z}\right) ,\qquad \mbox{where } X:U\subseteq \nc^n\to \nc^n.
\end{equation}
Conserved quantities and solution curve foliations are defined similarly to their real-valued counterparts. Indeed, a \textbf{first integral} of \eqref{DS} is a function $F:U\to \nc$ constant along every solution of \eqref{DS}.
And for every $\bm{z}\in U$, the unique solution $\varphi\p{t,\bm{z}}$ of \eqref{DS} such that
$\varphi\p{0,\bm{z}}=\bm{z}$ allows us to define a function $\varphi\p{\cdot,\cdot}$ in $n+1$ variables called the \textbf{flow} of \eqref{DS}.
Clarifying preliminary comments are in order whenever a particular solution $\bm{\psi}\p{t}$ is considered:
\begin{enumerate}
\item partial derivatives $\frac{\partial^k}{\partial\bm{z}^k}\varphi\p{t,\bm{\psi}}$  are multilinear functions of increasing order (or multidimensional matrices, see e.g. \cite{GelfandKapranovZelevinsky}) and appear in the Taylor series of the flow along $\bm\psi$:
\begin{equation}\label{taylorflow}
\varphi\p{t,\bm{z}} = \varphi\p{t,\bm{\psi}} + \frac{\partial \varphi\p{t,\bm{\psi}}}{\partial\bm{z}} \brr{\bm{z}-\bm{\psi}}+\frac{1}{2!}\frac{\partial^2\varphi\p{t,\bm{\psi}}}{\partial\bm{z}^2} \brr{\bm{z}-\bm{\psi}}^2+ \dots;
\end{equation}
\item each of these derivatives $\frac{\partial^k}{\partial\bm{z}^k}\varphi\p{t,\bm{\psi}}$ satisfies an echeloned set of differential systems, depending on the previous $k-1$ partial derivatives and customarily called \textbf{variational equations or systems}. They are explicitly called \textbf{higher-order} whenever $k\ge 2$.
\item variational system for $k=1$ is \emph{linear} and satisfied by the linear part of the flow along $\psi$:
\begin{equation}
\label{VE}
\tag{$\mathrm{VE}_{\psi}$}\dot{Y_1} =  A_1 Y_1, \qquad A_1\p{t}  :=  X'\p{\bm{\psi}}\in \mathrm{Mat}_n\p{K},
\end{equation}
$K=\nc\p{\psi}$ being the smallest differential field containing $\nc\p{t}$ and the solution.
\item For $k\ge 2$, however, the system is not linear, yet a linearised version may be found. \emph{The aim of the present paper is to do so with explicit formulae.}
\end{enumerate}

\subsection{Morales-Ramis-Ziglin theory and extensions} \label{moralesramisziglin}

Heuristics of all results within the Ziglin-Morales-Ramis-Sim\'o theoretical framework are firmly rooted in the following principle, expected to affect a widespread class of systems:
\begin{quote}
\emph{If general system \eqref{DS} is ``integrable" in some reasonable sense, then the  system satisfied by each of the partial derivatives of the flow at every particular solution $\bm{\psi}$ of \eqref{DS} must be also integrable in an accordingly reasonable sense.}
\end{quote}
Any attempt at ad-hoc formulations of this heuristic principle has an asset and a drawback:
\begin{itemize}
\item there \emph{is} a valid integrability axiom for linear systems, e.g. \eqref{VE}: the solvability of the Zariski identity component of the (linear algebraic) \emph{differential Galois group} \cite{Mo99a,SingerVanderput};
\item an explicit incarnation of this principle requires a clear notion of ``integrability" for \eqref{DS}.
\end{itemize}
The latter item is cleared in the Hamiltonian case by the \emph{Liouville-Arnold Theorem} establishing a sufficient condition for
a system to admit, at least locally, a new set of variables rendering it integrable by quadratures. Said condition is the hypothesis on $H$ in the following:
\begin{theorem}[Morales-Ruiz, Ramis, 2001] \label{moralesramis} Let $X_H$ be an $n$-degree-of-freedom Hamiltonian system
having $n$ independent first integrals in pairwise involution, defined on a neighborhood of an integral curve
$\psi $. Then, Galoisian identity component $\mathrm{Gal}\p{\mathrm{VE}_\psi}^\circ$ is an
abelian group. $\hfill\square$
\end{theorem}

See \cite[Cor. 8]{MoralesRamis} or \cite[Th. 4.1]{Mo99a} for a precise statement and a proof.

\begin{theorem} [Morales-Ruiz, Ramis, Sim\'o, 2005, \protect{\cite[Th. 5]{MoRaSi07a}}] \label{moralesramissimo}
Let $H$ be as in the previous theorem. Let $G_k$ be the differential Galois group of the $k$-th variational equations $\mathrm{VE}^k_\psi$, $k\ge 1$, and $G:=\varprojlim G_k$
the formal differential Galois group (inverse limit of the groups) of $X_H$ along $\psi$.
Then, the identity components of the Galois groups $G_k$ and $G$ are abelian. $\hfill\square$
\end{theorem}

Theorem \ref{moralesramissimo} makes use of  the language of jets,
after proving non-linear $\mathrm{VE}^k_\psi$ equivalent to \emph{any} consistent linearised completion. Efforts towards a constructive version of this main Theorem, as well as the line of study described in \S \ref{firstintegrals}, are hampered by a lack of consensus on the explicit block structure of this completion. The present work, summarised in its main result (Proposition \ref{LVEprop}) aims at contributing to fill in this gap. Hence, outcomes will be restricted to symbolic calculus
and bear no new results in the above theoretical framework.

\begin{notation}\rm \label{notalex} Part of the conventions listed below were already introduced in \cite{ABSW}.
\begin{enumerate}
\item[\textbf{1.}] The \emph{modulus} $i=\left|\mathbf{i}\right|$ of a multi-index $\mathbf{i}=\left(i_1 ,\ldots , i_n\right)\in\nz^n$ is the sum of its entries.
Multi-index addition and subtraction are defined entrywise as usual.
\item[\textbf{2.}] \emph{Multi-index order}:  $\p{i_1,\ldots,i_n}\le \p{j_1,\ldots,j_n}$ means
$ i_k \le j_k$ for every $k\ge 1$.
\item[\textbf{3.}] \emph{Standard lexicographic order}:  $\p{i_1,\ldots,i_n}<_{\mathrm{lex}}\p{j_1,\ldots,j_n}$ if
$ i_1 = j_1,\ldots,i_{k-1}=j_{k-1}$ and $i_k < j_k$ for some $k\ge 1$. \item[\textbf{4.}]
Given complex analytic $F\,: \, U\subset\nc^{n}\,\rightarrow\, \nc$ we define the
	{\em lexicographically sifted differential of $F$ of order $m$}   as the row vector
$
F^{(m)}\left(\bm x\right):=\mathrm{lex}\left(\frac{\partial^{m} F}{\partial x^{i_1}_1 \ldots \partial x^{i_n}_n }\left(\bm x\right)\right), $
where $ \av{\mb{i}} = m$ and entries are ordered as per $<_{\mathrm{lex}}$ on multi-indices.
\item[\textbf{5.}] We define
$
d_{n,k} := \binom{n+k-1}{n-1} , \; D_{n,k} := \sum_{i=1}^k d_{n,i}.
$
It is easy to check there are $d_{n,k}$ $k$-ples of integers in $\brr{1,\dots,n}$, and just as many
homogeneous monomials of degree $n$ in $k$ variables. \end{enumerate}
\end{notation}

\begin{notation}\rm Given integers $k_1,\dots,k_n\ge 0$, we define the usual multinomial coefficient as
\vspace*{4pt}
{\small\[
\binom{k_1+\dots+k_n}{k_1,\dots,k_n}:=\binom{k_1+\dots+k_n}{\mb{k}}:=\frac{\p{k_1+\dots+k_n}!}{k_1!k_2!\cdots k_n!}.
\]}

\vspace*{4pt}
\noindent
For a multi-index $\mb{k}\in\nz_{\ge 0}^n$,  define $\mb{k!}:=k_1!\cdots k_n!$. For any two such $\mb{k},\mb{j}$,
we define
\vspace*{4pt}
{\small\begin{equation} \label{newbinom} \binom{\mathbf{k}}{\mathbf{p}} := \frac{k_1!k_2!\cdots k_n!}{p_1!p_2!\cdots p_n!\p{k_1-p_1}!\p{k_2-p_2}!\cdots \p{k_n-p_n}!}
=  \binom{k_1}{p_1} \binom{k_2}{p_2} \cdots \binom{k_n}{p_n},
\end{equation}}

\vspace*{4pt}
\noindent and the multi-index counterpart to the multinomial,
$
\binom{\mb{k}_1+\dots+\mb{k}_m}{\mb{k}_1,\dots,\mb{k}_m}:=\frac{\p{\mb{k}_1+\dots+\mb{k}_n}!}{\mb{k}_1!\mb{k}_2!\cdots \mb{k}_n!}.
$
\end{notation}

\section{Symmetric products and powers of finite matrices}\label{symsection}

\subsection{Definition and properties}
The compact formulation called for by \eqref{taylorflow} and Notation \ref{notalex} (\textbf{3}) will be achieved through a product $\odot$ that was already defined by other means by U. Bekbaev (e.g.  \cite{bekbaevSept2009,bekbaevJan2010,bekbaevOct2010,bekbaevMar2012}) and will be systematised using basic categorical properties of the tensor product.
Let $K$ be a field and $V$ a $K$-vector space. See \cite{BlokhuisSeidel,Cartan,Lang} for details.
\begin{definition}
An \textbf{$r^{\mathrm{th}}$ symmetric tensor power} of $V$ is a vector space $S$, together with a symmetric multilinear map $\varphi:V^r:= V\times \stackrel{r}{\dots} \times V \to S$ satisfying the following \emph{universal property}: for every vector space $W$ and every symmetric multilinear map
$ f: V^r \to W $ there is a unique linear map
$ f_{\odot}: S \to W $
such that the following diagram commutes:
\[
\xymatrix{
V\times V\times \stackrel{r}{\dots} \times V \ar[d]_{\varphi} \ar[r]^{\mbox{\phantom{hhhhhh}}f} & W \\
S \ar@{-->}[ru]_{f_{\odot}} }
\]
In other words, $\mathrm{Hom}_K\p{S,W}\cong S\p{V^n,W}$ holds between the vector space of linear maps $S \to W$ and the vector space of symmetric multilinear maps $V^n\to W$.
\end{definition}
\begin{proposition} \label{symprop} Given any $K$-vector space $V$ and any $r\in \nn$,
\begin{enumerate}
\item a symmetric power $\p{\mathrm{Sym}^rV,\varphi}$ exists, unique up to isomorphism. We write $\bm{v}_1\odots \bm{v}_r:=\varphi\p{\bm{v}_1,\dots,\bm{v}_r}$,
$\bm{v}^{\odot k}:=\bm{v}\,\odot\stackrel{k}{\cdots}\odot \,\bm{v}$ for any  $\bm{v}\in V$, and
$\bm{v}^{\odot\mb{p}}:=\bm{v}^{\odot p_1}_1\odot\cdots\odot \bm{v}_n^{\odot p_n}$, for any $\bm{v}_1,\dots,\bm{v}_n\in V$ and multi-index $\mb{p}\in\nz^n_{\ge 0}$.
\item For any multilinear map $f:V^r\to W$, the linear map $f_{\odot}$ induced by the universal property is defined
on the generators of $\Sym^r V$ as
$ f_{\odot} \p{ \bm{v}_1\odot \cdots \odot \bm{v}_r} = f\p{ \bm{v}_1, \cdots , \bm{v}_r}. $
\item If $\dim_K V =n<\infty$ then every basis $\brr{\bm{e}_1,\dots , \bm{e}_n}$ of $V$ induces a basis for $\Sym^rV$:
\begin{equation} \label{Symbasis}
 \brr{  \p{\bm{e}_1\,\odot \stackrel{r_1}{\dots}\odot\, \bm{e}_1} \odot  \p{\bm{e}_2\,\odot \stackrel{r_2}{\dots}\odot \,\bm{e}_2} \,\odot \cdots \odot \, \p{\bm{e}_n\,\odot \stackrel{r_n}{\dots}\odot \,\bm{e}_n}  :r_i\ge 0,\,  \av{\mathbf{r}}=r } ;
 \end{equation}
hence, $ \dim_K \Sym^{r}V =  d_{n,r}$. Conventions $\Sym^1V=V$ and $\Sym^0V=K$ arise naturally.
\end{enumerate}
Hence,  product $\odot$ operates exactly like products of homogeneous polynomials in several variables.

\end{proposition}
 \begin{remark} $\Sym^r$ may also be defined in terms of the tensor power by $\Sym^rV=\bigotimes^rV/\sim$ modulo the relation
$\bm{v}_1\otimes\cdots\otimes\bm{v}_r\sim \bm{v}_{\sigma\p{1}}\otimes\cdots\otimes\bm{v}_{\sigma\p{r}}$, $ \sigma\in\fraks_r .$
\end{remark}
Given any $K$-vector space $W$ and two linear maps $f,g:V\to W$, define
\begin{equation} \label{universalh2}
h:V\times V \to \Sym^2 W, \qquad h\p{\bm{v}_1,\bm{v}_2} := \frac{1}{2}\qu{f\p{\bm{v}_1} \odot g\p{\bm{v}_2}+f\p{\bm{v}_2} \odot g\p{\bm{v}_1}}.
\end{equation} Immediately bilinear and symmetric, it is granted a unique linear
$ h_{\odot}: \Sym^2V \to \Sym^2 W $,
$ h_\odot \p{\bm{v}_1\odot \bm{v}_2} := h \p{\bm{v}_1, \bm{v}_2} $, by the universal property. Write $f \odot g:=h_{\odot}$. Then
$f \odot g = g \odot f$ and $ \p{f_1\circ f} \odot \p{g_1\circ g} = \p{f_1\odot g_1} \circ \p{f\odot g} $ for any linear maps $f_1,g_1:W\to W_1$.
A similar construction applies to the symmetric product of $m\ge 3$ linear maps $f_i:V\to W$:
\begin{equation}\label{universalhm}
\xymatrixrowsep{.1cm}
\xymatrix{\Sym^mV\ar[rr]^{f_1\odot \cdots \odot f_m} & &\Sym^m W \\
  \bm{v}_1\odots \bm{v}_m \ar@{|->}[rr] & & \frac1{m!}\sum_{\sigma\in\fraks_m}f_1\p{\bm{v}_{\sigma\p{1}}}\odots f_m\p{\bm{v}_{\sigma\p{m}}}.}
\end{equation}
Let us generalise the above symmetric product into one involving any two linear maps
\[ f: \Sym^{j_1} V \to \Sym^{i_1} W, \quad g: \Sym^{j_2} V \to \Sym^{i_2} W, \qquad j_1,j_2,i_1,i_2\geq 0. \]
Assume $V$ and $W$ finite-dimensional, $V$ having basis $\brr{\bm{e}_1,\dots,\bm{e}_n}$.
Defining the bilinear map
$\varphi \p{\bm{u}_1,\bm{u}_2} :=  \bm{u}_1\odot \bm{u}_2, \bm{u}_i\in\Sym^{j_i}V,
$
we are interested in finding a bilinear function $h$ in terms of $f$ and $g$ generalising \eqref{universalh2}, for which there is a unique linear $h_\odot$ completing the diagram
\begin{equation} \label{universalh}
\begin{tabular}{c}
\xymatrix{
 \Sym^{j_1} V\times  \Sym^{j_2} V \ar[d]_{\varphi} \ar[r]^{\mbox{\phantom{hhl}}h} &  \Sym^{i_1+i_2} W\\
 \Sym^{j_1+j_2} V \ar@{-->}[ru]_{h_\odot} }
 \end{tabular}
\end{equation}
We want $h$ to yield coefficient $1$ for all-round repeated vectors as in
\eqref{universalh2}. Symmetric, multilinear $\tilde h:V^{\times j_1+j_2} \to \Sym^{i_1+i_2} W$ is easier to define, generalising \eqref{universalh2}
and the example in \cite[p. 155]{Mo99a}: for any $\bm{u}_1,\dots, \bm{u}_{j_1+j_2}\in V$,
{\small\begin{equation}\label{htildefunc} \tilde{h}\p{\bm{u}_1,\dots, \bm{u}_{j_1+j_2}} \!:=\! \alpha_{j_1,j_2}
\sum f\p{\bm{u}_{\sigma\p{1}}\odot\!\cdots \!\odot\bm{u}_{\sigma\p{j_1}}\!} \! \odot  g\!\p{\bm{u}_{\sigma\p{j_1+1}}\odot\!\cdots\! \odot\bm{u}_{\sigma\p{j_1+j_2}}\!}\!, \end{equation}}where $ \alpha_{j_1,j_2}=\frac{1}{\binom{j_1+j_2}{j_1}}$ and the sum is taken over ${\sigma\in S_{j_1,j_2}}$ with
{\small\begin{equation}\label{Sj1j2}S_{j_1,j_2}:= \brr{\sigma\in\fraks_{j_1+j_2}:\sigma\p{1}<\dots<\sigma\p{j_1}\mbox{ and } \sigma\p{j_1+1}<\dots<\sigma\p{j_1+j_2}.}
\end{equation}}Define
$\p{\varphi_1\times\varphi_2}\p{\bm{u}_1,\dots, \bm{u}_{j_1+j_2}}=\p{\bm{u}_{i_1}\odots\bm{u}_{i_{j_1}},\bm{u}_{i_{j_1+1}}\odots\bm{u}_{i_{j_1+j_2}}},$
$\varphi_i$ being the universal map of $\Sym^{j_i}V$; we intend the diagram of functions involving the Cartesian product
\begin{equation}\label{mustcommute}
\begin{tabular}{c}\xymatrix{
V^{\times j_1+j_2} \ar[d]_{\varphi_1\times\varphi_2} \ar[dr]^{\tilde h} \\
 \Sym^{j_1} V\times  \Sym^{j_2} V  \ar[r]^{\mbox{\phantom{hhl}}h} &  \Sym^{i_1+i_2} W}
 \end{tabular}
\end{equation}
to commute. Let $\bm{u}_{i_1},\dots,\bm{u}_{i_{j_1+j_2}}\in \brr{\bm{e}_1,\dots,\bm{e}_n}$. Split into copies of separate basis\break vectors:
$ \brr{\bm{u}_{i_1},\dots,\bm{u}_{i_{j_1}}} =\brr{\bm{e}_1,\stackrel{p_1}{\dots},\bm{e}_1,\dots,\bm{e}_n,\stackrel{p_n}{\dots},\bm{e}_n}$,
$\brr{\bm{u}_{i_{j_1+1}},\dots,\bm{u}_{i_{j_1+j_2}}} =\break \brr{\bm{e}_1^{\times q_1},\dots,\bm{e}_n^{\times q_n}},$  with
$\av{\mb{p}}=j_1$ and $\av{\mb{q}}=j_2$, and define $\mb{k}=\mb{p}+\mb{q}$. The expression of \eqref{htildefunc} in these basis elements is now
an immediate consequence of basic combinatorics:
\[
\tilde h\p{\bm{e}_1\stackrel{k_1}{\dots},\bm{e}_1,\dots,\bm{e}_n,\stackrel{k_n}{\dots},\bm{e}_n}=\frac{1}{\binom{j_1+j_2}{j_1}}\sum_{\av{\mb{P}}=j_1,\mb{P}\le \mb{k}} \qu{\prod_{i=1}^n\binom{k_i}{P_i}}   f\p{\bm{e}^{\odot\mb{P}}} \odot g\p{\bm{e}^{\odot\mb{k}-\mb{P}}},
\]
leaving no option for \eqref{mustcommute} to commute but
\[
h\p{\bm{e}^{\odot\mb{p}},\bm{e}^{\odot\mb{q}}}=\frac{1}{\binom{j_1+j_2}{j_1}}\sum_{\av{\mb{P}}=j_1,\mb{P}\le \mb{p+q}}
\qu{\prod_{i=1}^n\binom{p_i+q_i}{P_i}  } f\p{\bm{e}^{\odot\mb{P}}} \odot g\p{\bm{e}^{\odot\mb{p}+\mb{q}-\mb{P}}}.
\]
Finally, the universal property on $\p{\Sym^{j_1+j_2} V,\widetilde\varphi}$ yields a unique $h_{\odot}$ such that
$h_\odot\circ \widetilde\varphi \equiv \tilde h$,
\begin{equation}\label{otherdiagram}
\begin{tabular}{c}
\xymatrix{
V^{\times j_1+j_2} \ar[d]_{\varphi_1\times\varphi_2} \ar[dr]^{\tilde h} \ar@{->}@/^-6pc/[dd]^{\widetilde\varphi} \\
 \Sym^{j_1} V\times  \Sym^{j_2} V \ar[r]^{\mbox{\phantom{hhl}}h}  \ar[d]^{\varphi} &  \Sym^{i_1+i_2} W\\
 \Sym^{j_1+j_2} V \ar@{-->}[ru]^{h_\odot} }
 \end{tabular}
\end{equation}
and $\varphi\circ\p{\varphi_1\times\varphi_2}\equiv \widetilde\varphi$. Fixing $\varphi$ (and $h$) the uniqueness of $h_\odot$ follows
from construction: any other $h_\bullet$ rendering \eqref{universalh} commutative would require the commutativity of the outer perimeter
of \eqref{otherdiagram}, hence $h_\bullet\equiv h_\odot$.
Hence all we need to do is express $f\odot g:=h_\odot$ in terms of its action on base elements \eqref{Symbasis} to obtain a simple, explicit form.
\begin{notation}When dealing with matrix sets, we will use super-indices and subind\-ices:
\begin{enumerate}
\item[\textbf{1.}] The space of \textbf{$\p{i,j}$-matrices} $\Mat_{m,n}^{i,j}\p{K}$ is either defined by its underlying set,\break i.e. all $d_{m,i}\times d_{n,j}$ matrices having entries in $K$, or as vector space\break $\Hom_K\p{\Sym^jK^m;\Sym^iK^n}$.
\item[\textbf{2.}] It is clear from the above that $\Mat_{n}^{0,0}\p{K}$ is the set of all scalars $\alpha\in K$ and $\Mat_{n}^{0,k}\p{K}$ (resp. $\Mat_{n}^{k,0}\p{K}$) is made up of all row (resp. column) vectors whose entries are indexed by $d_{n,k}$ lexicographically ordered $k$-tuples.
\item[\textbf{2.}] Reference to $K$ may be dropped and notation may be abridged if dimensions are repeated or trivial, e.g. $\Mat_{n}^{i,j}:=\Mat_{n,n}^{i,j}$, $\Mat_{m,n}^{i}:=\Mat_{m,n}^{i,i}$, $\Mat_n:=\Mat_n^{1}$, etcetera.
\end{enumerate}
\end{notation}
\noindent Checking product $\odot$ defined below renders diagrams \eqref{universalh} and \eqref{otherdiagram} commutative is immediate.
\begin{definition}[Symmetric product of finite matrices]\label{defSymProd}
Let $A\in \Mat_{m,n}^{i_1,j_1}\p{K}$, $B\in \Mat_{m,n}^{i_2,j_2}\p{K}$, i.e. linear maps
$A:\Sym^{j_1}K^{n}\to \Sym^{i_1}K^{m}$ and $B:\Sym^{j_2}K^{n}\to \Sym^{i_2}K^{m}$.
Given any multi-index $\mathbf{k}=\p{k_1,\dots,k_n}\in\nz^n_{\ge 0}$ such that $\av{\mathbf{k}} = k_1+\dots + k_n = j_1+j_2$,
define $C:= A\odot B  \in \Mat_{m,n}^{i_1+i_2,j_1+j_2}$ by
{\small\begin{equation} \label{SymProd}
C \p{\bm{e}_1^{\odot k_1}\cdots \bm{e}_n^{\odot k_n}} = \frac{1}{\binom{j_1+j_2}{j_1}}\sum_{ \mb{p}} \binom{\mathbf{k}}{\mathbf{p}} A\p{ \bm{e}_1^{\odot p_1}\cdots \bm{e}_n^{\odot p_n}} \!\odot \!
B\p{ \bm{e}_1^{\odot k_1-p_1}\cdots \bm{e}_n^{\odot k_n-p_n}},
\end{equation}}notation abused by removing $\odot$ to reduce space within basis elements \eqref{Symbasis}, binomials as in \eqref{newbinom}
and summation taking place for specific multi-indices $\mathbf{p}$, namely those such that
\[ \av{\mathbf{p}} = j_1 \qquad \mbox{ and \qquad } 0 \le p_i\le k_i, \quad i=1, \dots , n. \]
\end{definition}

The following is a mere exercise in induction:
\begin{lemma} Defining $\bigodot_{i=1}^rA_i$ recursively by $\p{\bigodot_{i=1}^{r-1}A_i}\odot A_r$ with
$A_i\in\Mat_{m,n}^{k_i,j_i}$,
\begin{equation}  \label{multipleproduct}
\p{A_1\odots A_r}\bm{e}^{\odot\mb{k}} = \frac{1}{\binom{j_1+\dots+j_r}{j_1,j_2,\dots,j_r}}\sum_{\mb{p}_1,\dots,\mb{p}_r}\binom{\mb{k}}{\mb{p}_1,\dots,\mb{p}_r}\bigodot_{i=1}^r A_i \bm{e}^{\odot\mb{p}_i},
\end{equation}
if $ \av{\mb{k}}=j_1+\dots+j_r$, sums obviously taken for $\mb{p}_1+\dots+\mb{p}_r=\mb{k}$ and $\av{\mb{p}_i}=j_i$, for every $i=1,\dots,r$. $\hfill\square$
\end{lemma}

\begin{remarks}
\begin{enumerate}
\item[\phantom{hola}]
\item[\textbf{1.}] For an equivalent ``non-monic'' formulation of \eqref{SymProd}  (i.e. one for which entry ${}_{1,1}$ need not have coefficient $1$) using multi-indices in both columns \emph{and} rows,
see e.g. \cite{bekbaevSept2009,bekbaevJan2010,bekbaevOct2010,bekbaevMar2012}.
\item[\textbf{2.}] Notation in Proposition \ref{symprop} extends to matrices: $\Sym^rA:=A^{\odot r} := A\stackrel{r}{\odots} A$.
\item[\textbf{3.}] For square $A\in\Mat^{1,1}_n$, powers ${}^{\odot r}$ according to \eqref{SymProd} and
\eqref{multipleproduct} are obviously consistent with multiple product \eqref{universalhm}, hence equal to established definitions for group morphism $\Sym^r:\gl_n\p{V}\to\gl_n\p{\Sym^r\p{V}}$ in multilinear algebra textbooks such as
the expression in terms of the \emph{permanent} of $A$ (e.g. \cite[Th. 9.2]{BlokhuisSeidel}),
or  $\frac1{r!}A\,\circledS\stackrel{r}{\cdots}\circledS\,A$ in \cite{Ap10a,ABSW,Ba99a}.
\end{enumerate}
\end{remarks}
\begin{example} Given matrices $A\in \Mat_{2}^{1,1}\p{K}$ and $B\in\Mat_{2}^{3,2}\p{K}$, we may write them as
{\small\[
A = \p{\begin{tabular}{c|c} $A\bm{e}_1$ &  $A\bm{e}_2$ \end{tabular}} =\p{ a_{ij}}_{i\le 2, j\le 2}, \quad B=
\p{\begin{tabular}{c|c|c} $B\bm{e}_1^{\odot 2}$ & $ B\bm{e}_1\odot\bm{e}_2$ & $  B\bm{e}_2^{\odot 2}$\end{tabular}}=\p{ b_{ij}}_{i\le 4, j\le 3},
\]}and it is immediate to check that the $\p{4,3}$ (hence four-column, five-row) matrix product
{\small $$ A\odot B = \p{\!\begin{tabular}{c|c|c|c} \!$\p{A\odot B}\p{\bm{e}_1^{\odot 3}}$
\! & \!$\p{A\odot B}\!\p{\bm{e}_1^{\odot 2}\odot\bm{e}_2}$ \! & \!$\p{A\odot B}\!\p{\bm{e}_1\odot\bm{e}_2^{\odot 2}}$ & $\p{A\odot B}\!\p{\bm{e}_2^{\odot 3}}$\!\end{tabular}\!} , $$}is equal to
$$\left(\!
\begin{array}{cccc}
 a_{11} b_{11} \!&\! \frac{a_{12} b_{11}+2 a_{11} b_{12}}{3}  \!&\! \frac{2 a_{12} b_{12}+a_{11} b_{13}}{3}  & a_{12} b_{13} \\
 a_{21} b_{11}+a_{11} b_{21} \!&\!
M_2  \!&\!
N_2  & a_{22} b_{13}+a_{12} b_{23} \\
 a_{21} b_{21}+a_{11} b_{31} \! &\!
M_3 \!&\!
N_3  & a_{22} b_{23}+a_{12} b_{33} \\
 a_{21} b_{31}+a_{11} b_{41} \! &\!
M_4  \!&\!
N_4  & a_{22} b_{33}+a_{12} b_{43} \\
 a_{21} b_{41} \! &\! \frac{a_{22} b_{41}+2 a_{21} b_{42}}{3}  \!&\! \frac{2 a_{22} b_{42}+a_{21} b_{43}}{3}  & a_{22} b_{43}
\end{array}
\!\right),$$ where $M_i,N_i$ are defined by
\begin{eqnarray*}
M_{i} &=& \frac{2 (a_{21} b_{{i-1},2}+a_{11} b_{i,2})+a_{22} b_{i-1,1}+a_{12} b_{i,1}}{3}, \\
N_{i} &=&  \frac{2 (a_{22} b_{i-1,2}+a_{12} b_{i,2})+a_{21} b_{i-1,3}+a_{11} b_{i,3}}{3}.
\end{eqnarray*}
\end{example}

The following is straightforward to prove from either direct application of the universal property or the techniques used in \cite{bekbaevSept2009,bekbaevOct2010}, and will not be delved into here:
\begin{proposition}\label{odotproperties}
For any $A$, $B$, $C$, and whenever products make sense,
\begin{enumerate}
\item $A\odot B = B\odot A$.
\item $\p{A+B} \odot C = A \odot C + B \odot C$.
\item $\p{A\odot B} \odot C = A\odot \p{B\odot C}$.
\item $\p{\alpha A}\odot B = \alpha \p{A\odot B}$ for every $\alpha \in K$.
\item If $A$ is square and invertible, then $\p{A^{-1}}^{\odot k} = \p{A^{\odot k}}^{-1}$.
\item $A\odot B =0$ if and only if $A=0$ or $B=0$.
\item If $A$ is a square $(1,1)$-matrix, then
$ A\bm{v}_1 \odot A\bm{v}_2 \odot \cdots \odot A\bm{v}_m = A^{\odot m} \bm{v}_1 \odot \cdots \odot \bm{v}_m . $
\item If $\bm{v}$ is a column vector, then
$\p{A\odot B} \bm{v}^{\odot \p{p+q}} =  \p{A\bm{v}^{\odot p}} \odot \p{B\bm{v}^{\odot q}}$, $p,q\in \nz_{\ge 0} . $
\end{enumerate}
\end{proposition}
Universal property on \eqref{htildefunc} and diagram \eqref{otherdiagram} with different notation yields:
\begin{lemma} \label{factorABvwlem}
For any two matrices $A\in\Mat^{i,j}_n$ and $B\in\Mat^{p,q}_{n}$ and $\bm{v}_1,\dots,\bm{v}_{j+q}\in V$,
$\p{A\odot B} \p{\bm{v}_1\odots\bm{v}_{j+q}}$ is equal to
\begin{equation}\label{factorABvw}
\frac{1}{\binom{j+q}{q}}\sum_{\sigma\in S_{j,q}}  A\p{\bm{v}_{\sigma\p{1}}\odots\bm{v}_{\sigma\p{j}}}
\odot B\p{\bm{v}_{\sigma\p{j+1}}\odots\bm{v}_{\sigma\p{j+q}}}, \quad
\end{equation}
$S_{j,q}$ defined as in \eqref{Sj1j2}. $\hfill\square$
\end{lemma}
\begin{lemma} \label{eLemma} \begin{enumerate}
\item (\cite{bekbaevSept2009,bekbaevOct2010}) For any $A\in \Mat_n^{p,q}$, $B\in \Mat_n^{q,r}$, and $\bm{v}\in\Sym^jK^n$,
\begin{equation}\label{vAB}
\p{\bm{v}\odot A} B = \bm{v}\odot \p{AB}.
\end{equation}
\item If $\bm{e_i}$ are the columns of $\Id_n$,
$ \sum_{m=1}^n \p{\bm{e}_m \odot \Id_n^{\odot k-1}} \p{\bm{e}_m^T \odot \Id_n^{\odot k-1}} = \Id_n^{\odot k+1}. $
\end{enumerate}
\end{lemma}
\begin{proof}
\begin{enumerate}
\item It suffices to prove it for basis elements of $\Sym^r$: for any $\mb{k}$ such that $\av{\mb{k}}=r$,
\begin{equation} \label{vId} \p{\bm{v} \odot A} B\bm{e}^{\odot\mb{k}} = \bm{v} \odot \p{AB}\bm{e}^{\odot\mb{k}}.
\end{equation}
But this is immediate from equation \eqref{factorABvw} or the definition \eqref{SymProd} of $\odot$ itself.
\item Using the previous item and the associative property in Proposition \ref{odotproperties},
and the fact that $\bm{e}_m\odot\bm{e}_m^T\in \Mat^{2,2}_n$ is zero save for a $1$ in position ${}_{m,m}$,
{\small$$ \sum_{m=1}^n \p{\bm{e}_m \odot \Id_n^{\odot k-1}} \p{\bm{e}_m^T \odot \Id_n^{\odot k-1}}
=\qu{\sum_{m=1}^n \p{\bm{e}_m \odot \bm{e}_m^T}} \odot \Id_n^{\odot k-1} = \Id_n^{\odot 2}\odot \Id_n^{\odot k-1}. $$}
\end{enumerate}
\end{proof}

\subsection{More properties of $\odot$}
We need to generalise some of the properties in Proposition \ref{odotproperties} for later purposes.
Applying the universal property on \eqref{universalh2} (with $V:=\Sym^kK^n$) or
\eqref{htildefunc} (with $j_1=j_2=k$), followed by \eqref{SymProd} and \eqref{anyprods} for $m=2$,
as well as the universal property on \eqref{universalhm} (with $\bm{v}_i:=\bigodot_{i=1}^m A_i \bm{e}^{\odot\mb{p}_i}$) and prepending
{\small${\binom{\av{\mathbf{j}}}{\mathbf{j}}}^{-1}\sum_{\mb{p}_1,\dots,\mb{p}_m}\binom{\mb{k}}{\mb{p}_1,\dots,\mb{p}_m}$} as in \eqref{multipleproduct} for arbitrary $m$, we obtain
\begin{lemma}\label{anyprodslem}
Given square $A,B\in\Mat_n^{k,k}$ and matrices $X_i\in\Mat_n^{k,j_i},\,i=1,2$,
\begin{equation}\label{anyprods}
\p{A\odot B}\p{X_1\odot X_2} = \frac{1}{2} \p{AX_1\odot BX_2 + BX_1\odot AX_2},
\end{equation}
and in general for any square $A_1,\dots,A_m\in\Mat_n^{k,k}$ and any $X_i\in\Mat_n^{k,j_i}$, $i=1,\dots,m$,
\begin{equation}\label{anyprodsgeneral}
\p{\bigodot_{i=1}^m A_i}\p{\bigodot_{i=1}^m X_i} = \frac{1}{m!} \sum_{\sigma\in\fraks_k}\bigodot_{i=1}^mA_{\sigma\p{i}}X_i. \qquad \square
\end{equation}
\end{lemma}
\noindent Defining $B:=\Id^{\odot m-j}_{n}$ and $\bm{v}_i:=X_i\bm{e}^{\odot\mb{p}_i}$ where $\av{\mb{p}_i}=q_i$ in equation \eqref{factorABvw}, we have:
\begin{lemma} \label{prodsmoregenerallem}
Given $A\in\Mat_n^{1,j}$ and $X_1,\dots,X_m$ such that $X_i\in\Mat_n^{1,q_i}$, $1\le j \le m$,
{\small\begin{equation}\label{prodsmoregeneral}
\binom{m}{j} \!\p{A\odot\Id^{\odot m-j}_{n}}\!\bigodot_{i=1}^m X_i = \sum_{1\le i_1<\dots<i_j\le m}\!\qu{A\p{X_{i_1}\odots X_{i_j}}}\odot \!\bigodot_{s\ne i_1,\dots,i_j} X_s. \quad \square
\end{equation} }
\end{lemma}
An immediate consequence of either Lemma \ref{anyprodslem} or Lemma \ref{prodsmoregenerallem} is
\begin{corollary} Given a square matrix $A\in\Mat_n^{1,1}$ and $X_1,\dots,X_m$ such that $X_i\in\Mat_n^{1,j_i}$,
\begin{equation}\label{prodsgeneral}
\p{A\odot\Id^{\odot m-1}_{n}}\p{\bigodot_{i=1}^m X_i} = \frac{1}{m} \sum_{i=1}^m\p{AX_i}\odot \p{X_1\odots \widehat{X_i}\odots X_m}. \qquad \square
\end{equation}
\end{corollary}
\noindent$ \p{X \bm{v} \odot \Id_n^{\odot r}}X^{\odot r} = X\bm{v} \odot X^{\odot r}$ in virtue of \eqref{vAB}; applying this, \eqref{vId}, Proposition \ref{odotproperties} and a detailed scrutiny of the effect on basis products $\bm{e}^{\odot\mb{\star}}$ yields:
\begin{lemma} \label{XYvBlem}
Given square matrix $X\in \Mat_n^{1,1}$, any vector $\bm{v}\in K^n$ and $r\ge 1$,
\begin{equation} \label{XYvB}
\p{X\bm{v}\odot \Id^{\odot r}}X^{\odot r} = X^{\odot r+1} \p{\bm{v}\odot \Id^{\odot r} } . \qquad \square
\end{equation}
\end{lemma}
\noindent If $\p{K,\partial}$ is a \emph{differential field} \cite{SingerVanderput} and we extend derivation $\partial$ entrywise, $\partial\p{a_{i,j}} := \p{\partial a_{i,j}}$, the Leibniz rule holds on \emph{vector} products $\bm{x}\odot\bm{y}$ as trivially as it does for homogeneous polynomials in $n$ variables in virtue of Proposition \ref{symprop} or $\Sym^{k_1+k_2}\p{V^\star}\cong S^{k_1+k_2}\p{V,K}$; \eqref{SymProd} implies:
\begin{lemma} \label{leibnizlem}
For any given
$X\in\Mat_n^{k_1,j_1}\p{K}$ and $Y\in\Mat_n^{k_2,j_2}\p{K}$,
\begin{equation}\label{leibniz}
\partial \p{X\odot Y} = \partial\p{ X} \odot Y + X\odot \partial\p{Y}.\qquad \square
\end{equation}
\end{lemma}

Although the next result will be rendered academic by simplified expressions in \S \ref{structure}, it is worth writing for the sake of clarifying certain routinely-appearing matrices a bit further. The proof is immediate from commutativity and \eqref{anyprodsgeneral}, \eqref{prodsgeneral}, Lemma \ref{leibnizlem}, the distributive property and \eqref{prodsgeneral}, as well as simple induction in (c):
\begin{lemma}\label{odeprods} Let $\p{K,\partial}$ be a differential field.
\begin{enumerate}
\item If $Y$ is a square $n\times n$ matrix having entries in $K$ and $\partial Y = AY$, then
\begin{equation} \label{Symkeq} \partial\; \Sym^k Y= k \p{A \odot \Sym^{k-1}\p{\Id_n}}\Sym^k Y.
\end{equation}
\item If $X\in \Mat_n^{1,j_1}$ and $Y\in \Mat_n^{1,j_2}$ satisfy systems $\partial X = AX+B_1$ and $\partial Y = AY+B_2$
with $A\in\Mat_n^{1,1},B_i\in\Mat_n^{1,j_i},
$
then symmetric product $X\odot Y$ satisfies linear system
\begin{equation} \label{Symskeq}
\partial \p{X\odot Y } = 2\p{A\odot \Id_{d_{n,k}}}\p{X\odot Y} +  \p{B_1\odot Y + B_2\odot X}.
\end{equation}
\item If $
\partial X_i = AX_i+B_i, \;  i=1,\dots,m,
$ with $X_i,B_i\in \Mat_n^{1,j_i}$,  $A\in \Mat_n^{1,1}$
then
\begin{equation} \label{anySymskeq}
\partial \bigodot_{i=1}^m X_i = m\p{A\odot\Id_{d_{n,k}}^{\odot m-1}}\bigodot_{i=1}^m X_i + \sum_{i=1}^mB_i\odot\bigodot_{j\ne i} X_j . \qquad \square
\end{equation}
\end{enumerate}
\end{lemma}
\begin{remark} \label{Liepower} Albeit not explicitly as in \eqref{Symkeq}, the matrix proven equal to \allowbreak $k (A \odot \Id_n^{\odot k-1})$  has appeared in numerous references (e.g. \cite{Ap10a,AW1,AW2,ABSW,Ba99a}) whenever a differential equation for $\Sym^k$ arises, has been sometimes labelled $\sym^k$ and has been consistently called symmetric power \emph{in the sense of Lie algebras}, its \emph{Lie group} counterpart therein equal to ${}^{\odot k}$ as defined in this paper. \end{remark}

\section{Symmetric products and exponentials of infinite matrices} \label{syminfsection}

The next step towards a compact form to linearised higher variationals is assembling the matrix blocks alluded to in Lemma \ref{odeprods} and Remark \ref{Liepower} together into a single matrix. Again, we follow paths already trod with other aims and formulations, e.g. by Bekbaev in \cite{bekbaevOct2010}.
\subsection{Products and exponentials}
Of the myriad ways to note a set of infinite matrices, we may need one taking finite submatrix orders into account. Alternatively,
of all the ways in which to write a $K$-algebra $S$, a need may arise to express it whenever possible $S=\Sym\p{V}:=\bigoplus_{k\ge 0} \Sym^k\p{V}$ for a given vector space.
\begin{notation}
Let $\Mat^{n,m}\p{K}$ denote the set of block matrices
$A =\p{A_{i,j}}_{i,j\ge 0}$ with $A_{i,j}:\Sym^{i}K^m\to \Sym^{j}K^n,$ hence $A_{i,j}\in \mathrm{M}_{d_{n,i}\times d_{m,j}} \p{K} = \mathrm{Mat}^{i,j}_{n,m}\p{K}$:
{\small\[
A =
\p{\begin{tabular}{>{\centering\arraybackslash}m{.5cm}|>{\centering\arraybackslash}m{2.2cm}|>{\centering\arraybackslash}m{1.2cm}|>{\centering\arraybackslash}m{-0.3cm}|c}
$\ddots$  &$\vdots$  & $\vdots$ & $\vdots$ &  \\ \cline{1-5}
\multirow{5}{*}{$\cdots$}    & \multirow{5}{*}{$A_{2,2}$} & \multirow{5}{*}{$A_{2,1}$} & \multirow{5}{*}{}&\multirow{5}{*}{$\hspace{-0.45cm}\leftarrow A_{2,0}$} \\
  &  &  &  &  \\
    &  &  &  &  \\
      &  &  &  &  \\
        &  &  &  &  \\  \cline{1-5}
\multirow{3}{*}{$\cdots$}    & \multirow{3}{*}{$A_{1,2}$} & \multirow{3}{*}{$A_{1,1}$} & \multirow{3}{*}{}&\multirow{3}{*}{$\hspace{-0.45cm}\leftarrow A_{1,0}$} \\
  &  &  &  &  \\
    &  &  &  &  \\  \cline{1-5}
$\cdots$    & $A_{0,2}$ & $A_{0,1}$ & & ${\hspace{-0.45cm}\leftarrow A_{0,0}}$ \\
\cline{1-5}
\end{tabular}
}
\]}We write $\Mat:=\Mat^{n,n}$ if $n$ is unambiguous.
Conversely, $\mathrm{Mat}^{i,j}_{n,m}$ is embedded in $\Mat^{n,m}$ by identifying every matrix $A_{i,j}$ with an element of $\Mat^{n,m}$ equal to $0$ save for block $A_{i,j}$.
\end{notation}

We define a product on $\Mat^{n,m}$. For a formulation yielding the same results see \cite[p. 2]{bekbaevMar2012}.
\begin{definition}
For any $A,B\in \Mat^{n,m}\p{K}$, define $A\odot B = C\in \Mat^{n,m}\p{K}$ by \begin{equation}\label{SymProdInf}
C=\p{C_{i,j}}_{i,j\ge 0}, \qquad C_{i,j} = \sum_{0\le i_1\le i,\; 0\le j_1\le j} \binom{j}{j_1}  A_{i_1,j_1}\odot B_{i-i_1,j-j_1}.
\end{equation}
Same as always, ${}^{\odot k}$ will stand for powers built with this product.
\end{definition}

The following is immediate and part of it has already been mentioned before, e.g. \cite{bekbaevOct2010}:
\begin{lemma} \label{Matisadomain} $\p{\Mat\p{K},+,\odot}$ is an integral domain, its identity element $1_{\Mat}$ equal to zero save for block $\p{1_{\Mat}}_{0,0}=1_K$.
$\Mat\p{K}$ is also a unital associative $K$-algebra with the usual product by scalars.
 $\hfill\square$
\end{lemma}

\begin{definition} (See also \cite{bekbaevOct2010}) \label{defexp}
for every matrix $A\in \Mat^{n,m}$ we define the formal power series
\[
\expodot A := 1 + A^{\odot 1} + \frac 12 A^{\odot 2} + \cdots = \sum_{i=0}^{\infty} \frac{1}{i!} A^{\odot i}.
\]
Whenever $A=0$ save for a finite distinguished submatrix $A_{j,k}$ (e.g. Examples \ref{examplesexp} below or
Lemma \ref{expfactor}), the abuse of notation $\expodot A_{j,k}=\expodot A$ will be customary.
\end{definition}
\noindent Commutativity of $\odot$ renders the proof of the following similar to that of scalar exponentials:
\begin{lemma} \label{expprops0}
\begin{enumerate}
\item[\phantom{hola}]
\item For every two $A,B\in\Mat^{n,m}$, $\expodot\p{A+B} =\expodot{A}\odot\expodot{B}$.
\item For every $Y\in \Mat^{n,m}$ and any derivation $\partial:K\to K$,  $\partial\expodot Y = \p{\partial Y}\odot \expodot Y. $
\item (\cite[Corollary 3]{bekbaevSept2009}) Given square matrices $A,B\in \Mat^{1,1}_n$, $\expodot AB = \expodot A\break \expodot B$.
\item In particular, for every invertible square $A\in \Mat^{1,1}_n$, $\expodot A^{-1} = \p{\expodot A}^{-1}$. \qed
\end{enumerate}
\end{lemma}


\begin{examples}\label{examplesexp}
\begin{enumerate}
\item[\phantom{hola}]
\item[\textbf{1.}] Let $A\in\Mat\p{K}$ such that all blocks are zero except for ${}_{1,1}$, the reader can check that
\[
\expodot A_{1,1}=\expodot\p{\begin{tabular}{c|c}
 $A_{1,1}$ & $0$ \\ \hline
 $ 0$ & $0$
\end{tabular}}=\mathrm{diag}\p{
\cdots , A_{1,1}^{\odot k} , \dots , A_{1,1} , 1} .
\]
\item[\textbf{2.}] If the only non-zero block in $A$ is a row vector,
\[
A=\p{\begin{tabular}{cccc|c}
$\cdots$ &  $0$ & $0$ &  $\cdots$ & $0$ \\ \hline
$\cdots$ &  $0$ & $\bm{x}$ &  $\cdots$ & $0$
\end{tabular}},\qquad \bm{x}=\p{x_1,\dots,x_{d_{n,k}}}\in  \mathrm{Mat}^{0,k}_n\p{K},
\]
the only non-zero block  in $A^{\odot j}$ is
$\p{A\odot A^{j-1}}_{0,jk}=\binom{jk}{k}\cdot \binom{\p{j-1}k}{k}\cdots \binom{2k}{k} A_{0,k}^{\odot j}$
 e.g. for $k=1$,
$
\expodot A= \sum_{j\ge 0}\bm{x}^{\odot j}  =
\p{\begin{tabular}{ccccc|c}
$\cdots$ &  $0$ & $0$ &$0$ &  $0$ & $0$ \\ \hline
$\cdots$ &  $\bm{x}^{\odot 4}$ & $\bm{x}^{\odot 3}$ & $\bm{x}^{\odot 2}$ &  $\bm{x}$ & $1$
\end{tabular}}.
$
\item[\textbf{3.}]
If the only non-zero block in $A$ is column ${}_{0,k}$, the only one in $A^{\odot j}$ is ${}_{jk,0}$, obtained by switching rows and columns and expunging binomials from \break$\p{A\odot A^{j-1}}_{0,jk}$ in \textbf{2}. For $k=1$,
{\small\[
\expodot \bm{x}=
\expodot \p{\begin{tabular}{c|c}
 $0$  & $\bm{x}$ \\ \hline
  $0$ & $0$
\end{tabular}} =  \p{\begin{tabular}{c|c}
  $\vdots$ &  $\vdots$  \\
  $0$ & $\frac{1}{j!}\bm{x}^{\odot j}$ \\
  $\vdots$ &  $\vdots$  \\
 $0$  & $\bm{x}$ \\ \hline
  $0$ & $1$
\end{tabular}}.
\]}
\end{enumerate}
\end{examples}
\noindent Fourth example \eqref{Ydisplay}, i.e. matrices equal to $0$ save for block row ${}_{1,k}$, deserves special attention.
\begin{notation}\label{notationc}
For every set of indices $1\le i_1\le \dots \le i_r$ such that $ \sum_{j=1}^ri_j = k$, $c^k_{i_1,\dots,i_r}$ is defined as the amount of totally ordered partitions of a set of $k$ elements among subsets of sizes $i_1,\dots,i_r$. We write $c^k_{\mathbf{i}}$ if $\mathbf{i}=\p{i_1,\dots,i_r}$ and omit super-index ${}^k$ if $\av{\mathbf{i}}$ is known beforehand.
\end{notation}
\begin{remarks} \begin{enumerate}
\item[\phantom{hola}]
\item[\textbf{1.}] $c^k_{i_1,\dots,i_j}=\# I^{i_1,\dots,i_j}_{1,\dots,k}$ in \eqref{Ims} below, $\sum_{\av{\mathbf{i}} = k} c^k_{i_1,\dots,i_j} = \stirlingtwo{k}{j},$  the \emph{Stirling number of the second type} (\cite[\S 24.1.4]{Abramowitz}), and  $\sum_{j=1}^k\sum_{\av{\mb{i}} = k} c^k_{i_1,\dots,i_j} = B_{k},$ the $k^{\mathrm{th}}$ \emph{Bell number} \cite[Vol 2, Ch. 3]{Ramanujan2}.
\item[\textbf{2.}]  Since each subset of size $i_s$ is supposed to be ordered, we must divide the total amount
by the orders of the corresponding symmetric groups, hence the explicit formula:
{\small\begin{equation} \label{ck}
c^k_{i_1,\dots,i_j} = \frac{\binom{k}{i_1\,i_2\,\cdots\,i_j}}{n_1!\cdots n_m!}, \quad
\left\{ \!\begin{array}{l}\p{i_1,\dots,i_j} = \p{k_1\stackrel{n_1}{\dots}k_1,\cdots,k_m\stackrel{n_m}{\dots}k_m}, \\ 1\le k_1<k_2<\dots< k_m.\end{array} \right.
\end{equation}}
\end{enumerate}
\end{remarks}

\begin{lemma}\label{AZproperties}
Let $Y\in\Mat\p{K}$ equal to zero outside of block row $_{1,k}$, $k\ge 1$:
\begin{equation} \label{Ydisplay}
Y:= \p{\begin{tabular}{cccc|c}
$\cdots$ &  $Y_3$ & $Y_2$ &  $Y_1$ & $0$ \\ \hline
$0$ &  $0$ & $0$ &  $0$ & $0$
\end{tabular}}, \qquad Y_i\in\Mat^{1,i}_{n}.
\end{equation}
Let $Z_{r,s}$, $s, r \ge 1$, be the corresponding block in $\expodot Y$. Then,
\begin{enumerate}
\item Row block $r$ in $\expodot Y$ is recursively obtained in terms of row blocks $1$ and $r-1$:
\begin{equation}\label{expZrs}
Z_{r,s} = \frac1r \sum_{j=1}^{s-r+1}\binom{s}{j}Y_j\odot Z_{r-1,s-j}.
\end{equation}
In particular, $Z_{r,r}=Y_1^{\odot r}$ and $Z_{r,s}=0_{d_{n,r},d_{n,s}}$ whenever $r>s$.
\item For every $m,r\ge 1$ and any $\bm{v}\in K^n$,
\begin{equation}\label{YZout}
\p{Y_1\bm{v}\odot \Id_n^{\odot r}} Z_{r,r} =  Z_{r+1,r+1}\p{\bm{v}\odot \Id_n^{\odot r}}.
\end{equation}
\item Using Notation \ref{notationc} and \eqref{ck}, for every $s\ge r$
\vspace*{2pt}
\begin{equation} \label{Zrkdef} Z_{r,s}=\sum_{i_1+\dots+i_r=s} c^s_{i_1,\dots,i_r}Y_{i_1}\odot Y_{i_2} \odot \cdots \odot Y_{i_r}.
\end{equation}
\item Let $A\in\Mat\p{K}$ similar to $Y$, its horizontal strip not necessarily at level ${}_{1,*}$:
\vspace*{2pt}\[
A:= \p{\begin{tabular}{cccc|c}
$\cdots$ &  $A_3$ & $A_2$ &  $A_1$ & $0$ \\
$\cdots$  &  $0$ & $0$ &  $0$ &  $0$  \\
  &  $\vdots$ & $\vdots$ &  $\vdots$ &  $\vdots$  \\
$\cdots$  &  $0$ & $0$ &  $0$ &  $0$  \\
\hline
$0$ &  $0$ & $0$ &  $0$ & $0$
\end{tabular}}, \quad A_t\in\Mat_n^{p,t}.
\]
For every $t,i\ge 1$ and $s\ge t+i$, the following factorization holds:
\vspace*{2pt}\begin{equation}\label{Afactor}
 \sum_{j=t}^{s-i}\binom{s}{j} \p{A_tZ_{t,j}}\odot Z_{i,s-j}=\binom{t+i}{i} \p{A_t\odot \Id_n^{\odot i}} Z_{t+i,s} .
\end{equation}
\item If $Q\in \Mat^n$ has only its square ${}_{1,1}$ block different from zero, then $\expodot QY = \p{\expodot Q}\p{\expodot Y}$.
\end{enumerate}
\end{lemma}
\begin{proof}
\begin{enumerate}
\item Using \eqref{SymProdInf} on $A=Y$, $B=Y^{\odot s-1}$,
as well as the fact $Z_{i,j}=0$ for $i>j$, \eqref{expZrs} ensues.
\item Direct from \eqref{XYvB} in Lemma \ref{XYvBlem}.
\item By induction. For $s=1$, $r$ can only be equal to $1$ in order to have a non-zero block and $Z_{1,1}=Y_1=c^{1}_{1}Y_1$.
Assume \eqref{Zrkdef} holds for all $r$ smaller than or equal to $s-1$. Summand redistribution renders $Z_{r,s} = \frac1r \sum_{j_1+\dots+j_r=s} C_{j_1,\dots,j_r}  Y_{j_1}\odot Y_{j_2} \odot \cdots \odot Y_{j_{r}} $
where $C_{j_1,\dots,j_r}$ splits into a certain sum, each of whose $m$ terms is easily checked to be equal to
$  n_m c^s_{j_1,\dots,j_r}, $
hence the coefficient of $Y_{j_1} \odot \cdots \odot Y_{j_{r}}$ in equals
$\frac1r \sum_{i=1}^mn_i c^s_{j_1,\dots,j_r} = c^s_{j_1,\dots,j_r}. $
\item the left-hand side in \eqref{Afactor}, expressed in terms of \eqref{Zrkdef} and applying distributivity, equals
\vspace*{2pt}{\small\begin{equation}\label{conydesuma}  \sum_{j=t}^{s-i}\binom{s}{j}
\sum_{m_1,\dots,m_t}\sum_{k_1,\dots,k_i}c^j_{m_1,\dots,m_t}c^{s-j}_{k_1,\dots,k_i}\qu{A_t\p{Y_{m_1}\odots Y_{m_t}}} \odot
Y_{k_1}\odots Y_{k_i}.
\end{equation}}A tedious exercise in counting index multiplicities and applying basic combinatorics allows us to apply Lemma \ref{prodsmoregenerallem} to $A_t$ and $Y_{\odot \mathbf{m}} \odot  Y_{\odot \mathbf{k}}:=Y_{m_1}\odots Y_{m_t} \odot  Y_{k_1}\odots Y_{k_i}$:
\vspace*{2pt}{\small\begin{equation}\label{betterequation}
\binom{s}{j}c^j_{\mathbf{m}}c^{s-j}_{\mathbf{k}}\qu{A_t\p{Y_{m_1}\odots Y_{m_t}}} \odot
Y_{k_1}\odots Y_{k_i} = \binom{i+t}{i}\p{A_t\odot \Id_n^{i}}Y_{\odot \mathbf{m}} \odot  Y_{\odot \mathbf{k}}.
\end{equation}}The fact every summand in \eqref{conydesuma} fits the same profile as the left-hand side in \eqref{betterequation} allows us to factor $\binom{i+t}{i}\p{A_t\odot \Id_n^{i}}$
out of the whole sum, namely $Z_{i+t,s}$.
\item Replacing each factor $Y_{i_j}$ by $QY_{i_j}$ in \eqref{Zrkdef} and applying Lemma \ref{odotproperties} we obtain
$\expodot QY = \p{\tilde Z_{r,k}}$ where
$ \tilde Z_{r,s}=\sum_{i_1+\dots+i_r=s} Q^{\odot r} \odot c^s_{i_1,\dots,i_r}Y_{i_1}\odot Y_{i_2} \odot \cdots \odot Y_{i_r} = Q^{\odot r} Z_{r,s}, $
hence matrix $\expodot Y$ appears multiplied by $\mathrm{diag}\p{ \cdots,Q^{\odot 2}, Q^{\odot 1}, 1}=\expodot Q$.
\end{enumerate}
\end{proof}

\begin{lemma}\label{expfactor}
Let $A$ and $Y$ be as in Lemma \ref{AZproperties}.
Then, \begin{equation}\label{Aexp} \p{A\expodot Y} \odot \expodot Y=\p{A\odot \expodot\Id_n} \expodot Y.\end{equation}
\end{lemma}
\begin{proof}
Based on \eqref{SymProdInf}, $B:=A\odot \expodot\Id_n\in\Mat\p{K}$ is defined recursively by
\vspace*{2pt}\[
B_1=A_1,  \qquad
B_k= \p{
\begin{tabular}{ccc}
$\binom{k}{k-1}A_1\odot \Id_n^{\odot {k-1}}$ & \\
\cline{2-3} $\binom{k}{k-2}A_2\odot \Id_n^{\odot k-2}$  & \multicolumn{2}{|c|}{}\\
\vdots  & \multicolumn{2}{|c|}{\phantom{a}$B_{k-1}$\phantom{a}} \\
$\binom{k}{0} A_k$  & \multicolumn{2}{|c|}{} \\\cline{2-3}
\end{tabular}
} , \quad k\ge 2.
\]
Let \[
\Phi_1=Y_1,\qquad \Phi_k =  \p{
\begin{tabular}{ccc}
$Z_{k,k}$ & \\
\cline{2-3} $Z_{k-1,k}$  & \multicolumn{2}{|c|}{}\\
\vdots  & \multicolumn{2}{|c|}{$\phantom{a}\Phi_{k-1}\phantom{a}$} \\
$Z_{1,k}$  & \multicolumn{2}{|c|}{} \\\cline{2-3}
\end{tabular}
}, \quad k\ge 2,\]
be the matrix formed by the first $k$ row and column blocks in $\expodot Y$.
Let $M_{r}$ be the block row $r$ of $B$, $A_k:= \p{A_{1,k},A_{1,k-1}\dots,A_{1,1}}$ the first $k$ blocks in $A$
and $Z^k:=\p{Z_{k,k},Z_{k-1,k}\dots,Z_{1,k}}^T$  the first block column in $\Phi_k$.
Given $s\ge 1$, block ${}_{1,s}$ in $A\expodot Y$ equals
$ A_s Z^s = \sum_{j=1}^s A_j Z_{j,s} ,$
hence for every $r=1,\dots,k$ block $_{r,k}$ in $\p{A\expodot Y} \odot \expodot Y$ can be rewritten, in virtue of \eqref{Afactor} with $p=1$, $s=k$ and $i=r-1$, as $ \sum_{t=1}^{k-r+1} \binom{t+r-1}{r-1} \p{A_t\odot \Id_n^{\odot r-1}} Z_{t+r-1,k}=M_r Z^k$.
\end{proof}

\subsection{Application to power series} \label{powerseriessection}
Since polynomials and power series split into homogeneous components, Example \ref{examplesexp}(\textbf{3}) implies:
\begin{lemma} \label{expformalseries}
\begin{enumerate}
\item[\phantom{hola}]
\item Let $F\in K\qu{\qu{\bm{x}}}$, $\bm{x}=\p{x_1,\dots,x_n}$, be a formal series.
Then there exists a set of row blocks $M_F^{1,i}\in\Mat_{m,n}^{1,i}\p{K}$, $i\ge 0$ such that $F$ admits the expression $F\p{\bm{x}} = M_F \expodot X$, where
\[
M_F := \p{\begin{tabular}{ccc|c}
$\cdots$  & $M_F^{1,2}$ &  $M_F^{1,1}$ & $M_F^{1,0}$ \\\hline
$\cdots$  & $0$ &  $0$ & $0$
\end{tabular}}\in\Mat^{1,n}\p{K}, \qquad X := \p{\begin{tabular}{c|c}
  $0$     & $\bm{x}$ \\ \hline
  $0$ & $0$
\end{tabular}}.
\]
\item If $F=F_1\times \dots \times F_m$ is a \emph{vector} power series, adequate $M_F^{1,i}\in\Mat_{m,n}^{1,i}\p{K}$ render
\[
\fbox{$F\p{\bm{x}} = M_F \expodot X$}\quad \mbox{ where }
M_F := \p{\begin{tabular}{ccc|c}
$\cdots$  & $M_F^{1,2}$ &  $M_F^{1,1}$ & $M_F^{1,0}$ \\\hline
$\cdots$  & $0$ &  $0$ & $0$
\end{tabular}}\in\Mat^{m,n}.
\]
\end{enumerate}
Following Definition \ref{defexp}, write $F\p{\bm{x}} = M_F \expodot \bm{x}$ if it poses no clarity issue. \qed
\end{lemma}
From the above Lemma it follows that every formal power series can be expressed in the form
$M_F\expodot{\bm{x}}$, where abusing notation once again
\begin{equation} \label{Fjet}
M_F = J_F+M_F^{1,0} := \p{\begin{tabular}{ccc|c}
$\cdots$  & $M_F^{1,2}$ &  $M_F^{1,1}$ & $0$ \\\hline
$\cdots$  & $0$ &  $0$ & $0$
\end{tabular}}
+  \p{\begin{tabular}{c|c}
   $0$ & $M_F^{1,0}$ \\\hline
   $0$ & $0$
\end{tabular}}.
\end{equation}
In other words: $M_F$ equals the sum of two matrices with easily computable $\odot$-exponentials: one following
Example \ref{examplesexp} (\textbf{3}) (same as \bm{x}) and one following \eqref{Ydisplay}.
 Lemma \ref{expprops0}, Lemma \ref{Matisadomain} and the universal property of finite products $\odot$ yield the following two results; see \cite{bekbaevSept2009,bekbaevOct2010} for a proof.
\begin{lemma} \label{lemaboutseries} Given power series $F=\p{F_1,\dots,F_m}$ and $G=\p{G_1,\dots,G_p}$ in $n$ and $m$ indeterminates, respectively,
\begin{enumerate}
\item If $n=m$, $M_{FG} = M_F\odot M_G$.
\item $\expodot F\p{\bm{x}} = \p{\expodot M_F} \p{\expodot \bm{x}}$.
\item $M_{G\circ F} = M_{G}\expodot M_F$.
\item $\expodot\p{M_G\expodot M_F} = \p{\expodot M_G} \p{\expodot M_F}$. $\hfill\square$
\end{enumerate}
\end{lemma}
\begin{corollary}\label{corchange}
Let $F\p{\bm{x}}=\p{F_1,\dots,F_p}\p{x_1,\dots,x_n}$ be a vector power series, $\bm{y} = F\p{\bm{x}}$ and
$$ \left.\begin{array}{ccl}\bm{X}&=&R_{x,X} \expodot \bm{x}\in K^N,\\ \bm{Y}&=&S_{y,Y} \expodot \bm{y}\in K^P,\end{array}\right\} \qquad R_{x,X}\in\Mat^{N,n}\p{K}, \quad S_{y,Y}\in\Mat^{P,p}\p{K}, $$
be independent and dependent variable changes, which we assume admit formal inverse changes
$$ \left.\begin{array}{lll}\bm{x}&=&R_{X,x} \expodot \bm{X},\\ \bm{y}&=&S_{Y,y} \expodot \bm{Y},\end{array}\right\} \qquad R_{X,x}\in\Mat^{n,N}\p{K}, \quad S_{Y,y}\in\Mat^{p,P}\p{K}. $$ Then, the expression of $F$ in the new variables, written in that
in those old, is
{\small\begin{equation} \label{newold} \fbox{$M_{F,X,Y} =  S_{y,Y}\p{\expodot M_{F,x,y}} \expodot R_{X,x} $} \quad \mbox{ where } \bm{y}=F\p{\bm{x}}=M_{F,x,y}\expodot  \bm{x}.
\; \square \end{equation}}
\end{corollary}
As was hinted at in \cite[p. 5]{bekbaevOct2010}, this result shows interesting light on the way finite-level transformations translate
into transformations on $\Mat^{n,m}$. For a linear transformation of the independent variables $\bm{x} = B\bm{X}$, however,
basic properties of $\expodot$ are as useful as \eqref{newold} in proving $F$ admits the following expression in the new variable $\bm X$
(mind the effect of the first matrix, equal to zero save for block ${}_{1,1}$ which is equal to $\Id_n$, on the second one):
\begin{equation}\label{Fxnew}
 F\p{\bm{X}} = \Id_n\p{\expodot{M_F}}\p{\expodot{B}} \bm X = \p{J_F+M^{0,0}_F}\p{\expodot{B}} \expodot \bm X .
\end{equation}
This will be applied to first integrals of dynamical systems in Section \ref{firstintegrals}.

\section{Higher-order variational equations}
\subsection{Structure}\label{structure}

Let us step back to what was said in \S \ref{intro1}.
For each particular integral curve $\bm\psi$ of a given
complex autonomous dynamical system \eqref{DS},
the \textbf{variational system} $\mathrm{VE}^k_{\psi}$ for \eqref{DS} along
$\bm\psi$ is satisfied by partial derivatives
$\frac{\partial^k}{\partial\bm{z}^k}\varphi\p{t,\bm \psi}$.
Case $k=1$ being trivial as shown in \eqref{VE}, the situation of interest is $k>1$.
We will eschew formulations such as those in \cite[eq (14)]{MoRaSi07a} in favour of the explicit formulae
\eqref{PhiLVE}, \eqref{ALVE}, \eqref{LVE} and \eqref{VEkredux} using Linear Algebra to express multilinear maps.
\begin{notation}\label{notaAY}
$K:=\nc\p{\bm\psi}$,
$A_i:=X^{\p{i}}\!\p{\bm{\psi}}$, $ Y_i := \mathrm{lex}\p{\frac{\partial^i }{\partial \bm{z}^i}\varphi{\p{t,\bm{\psi}}}}$
and, per Lemma \ref{AZproperties},
\begin{equation} \label{PhiLVE}
\Phi_1=Y_1,\qquad \Phi_k =  \p{
\begin{tabular}{ccc}
$Z_{k,k}$ & \\
\cline{2-3} $Z_{k-1,k}$  & \multicolumn{2}{|c|}{}\\
\vdots  & \multicolumn{2}{|c|}{$\phantom{a}\Phi_{k-1}\phantom{a}$} \\
$Z_{1,k}$  & \multicolumn{2}{|c|}{} \\\cline{2-3}
\end{tabular}
}, \quad k\ge 2,\end{equation}
formed by the first $k$ block rows and columns in $\Phi=\expodot Y$. Define $A,Y\in \Mat\p{K}$ as in Lemma \ref{AZproperties} with the above $A_i$, $Y_i$ as blocks.
Denote the canonical basis on $K^n$ (meaning the set of columns of $\Id_n$) by $\brr{\bm{e}_1,\dots,\bm{e}_n}$.
\end{notation}

\begin{lemma} \label{Zaux}
In the hypotheses described in Notation \ref{notaAY}, let $k\ge 1$ and $m=1,\dots,n$. Then,
\begin{eqnarray}
Y_k &=& \sum_{j=1}^n \frac{\partial Y_{k-1}}{\partial z_j} \p{\bm{e}_j^T \odot \Id_n^{\odot k-1}}, \label{Zaux1} 
\end{eqnarray}\begin{eqnarray}
\frac{\partial}{\partial z_m} Y_{k}&=& Y_{k+1} \p{\bm{e}_m \odot \Id_n^{\odot k}}, \label{Zaux2} \\
\frac{\partial}{\partial z_m} Z_{r,k} &=& Z_{r,k+1}\p{\bm{e}_m \odot \Id_n^{\odot k}} - \p{ Y_1\bm{e}_m\odot \Id_n^{\odot r-1} }Z_{r-1,k}, \quad r\le k, \label{Zaux3}\\
\frac{\partial}{\partial z_m} A_k &=& A_{k+1} \p{Y_1\bm{e}_m\odot \Id_n^{\odot k}} \label{Zaux4}.
\end{eqnarray}
\end{lemma}
\begin{proof}
We will explicitly prove \eqref{Zaux2}; \eqref{Zaux1} is an immediate consequence of Lemma \ref{eLemma} and \eqref{Zaux2}. We have, for every given ordered multi-index $\mb{i}=\p{i_1,\dots,i_k}$,
\[ \frac{\partial Y_{k}}{\partial z_{m}}\bm{e}_{i_1}\odots\bm{e}_{i_k} =\frac{\partial}{\partial z_{m}} \frac{\partial^k \varphi}{\partial z_{i_1}\partial z_{i_2}\cdots\partial z_{i_k}}  = Y_{k+1} \bm{e}_m\odot \bm{e}_{i_1}\odots\bm{e}_{i_k}.
\]
The right-hand side in \eqref{Zaux2} is equal to this expression, too, by simple application of the same principle as in \eqref{vId}.
The effect of $ \frac{\partial}{\partial \bm{z}}$ on $A_j$ is clear as well: chain rule implies
\[ \frac{\partial A_{k}}{\partial z_{m}}\bm{e}_{i_1}\odots\bm{e}_{i_k}=
\sum_{r=1}^n \frac{\partial^{k+1} X}{\partial z_{i_1}\partial z_{i_2}\cdots\partial z_{i_k}\partial z_r}
 \frac{\partial\varphi_r}{\partial z_{m}} =\sum_{r=1}^n A_{k+1}\p{\bm{e}_r\odot \bm{e}^{\odot\mb{i}}} \frac{\partial\varphi_r}{\partial z_{m}} , \]
which is equal, again using \eqref{vId} in order to obtain {\small$\frac{\partial\varphi_r}{\partial z_{m}}\bm{e}_r\odot \bm{e}^{\odot\mb{i}}=\p{\frac{\partial\varphi_r}{\partial z_{m}}
\bm{e}_r\odot \Id_n^{\odot r}}\bm{e}^{\odot\mb{i}}$},  to
{\small\[
A_{k+1}\sum_{r=1}^n \p{\bm{e}_r\odot \bm{e}^{\odot\mb{i}}} \frac{\partial\varphi_r}{\partial z_{m}}=A_{k+1}\sum_{r=1}^n \p{\frac{\partial\varphi_r}{\partial z_{m}}\bm{e}_r\odot \bm{e}^{\odot\mb{i}}} =A_{k+1}\sum_{r=1}^n \p{\frac{\partial\varphi_r}{\partial z_{m}}
\bm{e}_r\odot \Id_n^{\odot r}}\bm{e}^{\odot\mb{i}} ,
\]}hence to $  A_{k+1}  \p{ \frac{\partial \varphi}{\partial z_m} \odot \Id_n^{\odot k} }\bm{e}^{\odot\mb{i}} = A_{k+1}  \p{ Y_1\bm{e}_m \odot \Id_n^{\odot k} }\bm{e}^{\odot\mb{i}}.$

The reader can check \eqref{Zaux3} and \eqref{Zaux4}. For instance the latter is obtained by induction over $k$ using derivation of \eqref{expZrs},  \eqref{Zaux2} and Leibniz rule \eqref{leibniz}, as well as
application of \eqref{Afactor} with $i=1$, $t=r-2$, $s=k$, $A_t=Y_1\bm{e}_m\odot \Id_n^{\odot r-2}$  and $p=r-1$, use of \eqref{XYvB} and the fact $Z_{r-2,r-2}=Y_1^{\odot r-2}$.
\end{proof}

\begin{proposition} [First explicit version of non-linearised $\mathrm{VE}_{\psi}^k$]\label{nonlinearprove2}
In the above hypotheses,
\begin{equation} \label{VEphi} \tag{$\mathrm{VE}_\psi$}
\dot Y = A\expodot Y;
\end{equation}
in other words, for every $k\ge 1$,
\begin{equation} \label{VEkredux} \tag{$\mathrm{VE}^k_\psi$}
\frac{d}{dt}Y_k = \sum_{j=1}^k A_j Z_{k,j}= \sum_{j=1}^k A_j \sum_{i_1+\dots+i_j = k} c^k_{i_1,\dots,i_j}Y_{i_1}\odot Y_{i_2} \odot \cdots \odot Y_{i_j}.
\end{equation}
\end{proposition}
\begin{proof}
Assume  $\mathrm{VE}_{\psi}^{k-1}$ can be expressed as
$ \frac{d}{dt}Y_{k-1} = \sum_{j=1}^{k-1} A_j Z_{j,k-1}$.
The entries\break in $Y_{k-1}$ are partial derivatives of $\varphi\p{t,\bm{z}}$, hence $\frac{d}{dt}\equiv \frac{\partial}{\partial t}$ on every entry, Schwarz\break Lemma applies and derivation of \eqref{Zaux1} yields
$\frac{d}{dt} Y_k =
= \sum_{m=1}^n \frac{\partial }{\partial z_m}\break\frac{\partial Y_{k-1}}{\partial t} \p{\bm{e}_m^T \odot \Id_n^{\odot k-1}}
; $ induction hypothesis and Leibniz rule render $\frac{d}{dt} Y_k=
\sum_{m=1}^n \break\qu{\sum_{p=1}^{k-1}\frac{\partial A_p}{\partial z_m} Z_{p,k-1}+ A_p\frac{\partial Z_{p,k-1}}{\partial z_m} } \p{\bm{e}_m^T \odot \Id_n^{\odot k-1}}; $
equations \eqref{Zaux3} and \eqref{Zaux4} imply this is equal to $S_1+S_2-S_3$, where
\begin{eqnarray*}
S_1&=&\sum_{m=1}^n \sum_{p=1}^{k-1} A_{p+1} \p{Y_1\bm{e}_m\odot \Id_n^{\odot p}} Z_{p,k-1} \p{\bm{e}_m^T \odot \Id_n^{\odot k-1}}, \\
S_2 &=& \sum_{m=1}^n \sum_{p=1}^{k-1}A_pZ_{p,k}\p{\bm{e}_m \odot \Id_n^{\odot k-1}} \p{\bm{e}_m^T \odot \Id_n^{\odot k-1}}, \\
S_3 &=& \sum_{m=1}^n \sum_{p=1}^{k-1} A_p \p{Y_1\bm{e}_m\odot \Id_n^{\odot p-1} }Z_{p-1,k-1}  \p{\bm{e}_m^T \odot \Id_n^{\odot k-1}}.
\end{eqnarray*}
Sum swapping in $\sum_m\sum_p$ and Lemma \ref{eLemma} (b)  imply
$S_2  =  \sum_{p=1}^{k-1} A_pZ_{p,k}$;  \eqref{YZout}, Lemma \ref{eLemma} (b) and Proposition \ref{odotproperties} (b) render $S_1 - S_3$ equal to
missing summand $A_k Z_{k,k}$ in $S_2$.
\end{proof}
\begin{corollary}[Second explicit version of non-linearised $\mathrm{VE}_{\psi}^k$]
Let $\varphi\p{t,\bm{\psi}} = \p{\varphi_i}_i$ be the flow of \eqref{DS} along $\bm\psi$. Given
$k\ge 1$, $\mb{N}\in \nn^k$,$r=1,\dots,k$ and $0\le m_1\le \dots\le m_r$, define:
\begin{enumerate}
\item the set $S^{\mb m}$ of $\sigma\in\fraks_k$ such that $\sigma\p{1,\dots,k}=\p{\mathbf{i}_1,\dots,\mathbf{i}_r}$, $\mathbf{i}_j=\p{i_{j,1},\dots ,i_{j,m_j}}$ and $i_{j,s}<i_{j,s+1}$ for every $j,s$ and $m_j=m_{j+1}$ implies $i_{j,1}<i_{j+1,1}$;
\item the index-ordered partitions of $\mb{N}$ in  subsets of sizes
$0\le m_1\le \dots\le m_r$:
\begin{equation}\label{Ims} I^{\mb m}_{\mb N}:=\brr{\p{N_{\sigma\p{1}},\dots,N_{\sigma\p{k}}}=\p{\mathbf{K}_1,\dots,\mathbf{K}_r}: \mb{K}_i\in \nn^{m_i}\mbox{ and } \sigma\in S^{\mb m}}; \end{equation}
\item and, using abridged notation $\sum_{j_1,\dots,j_r}$ to denote $\sum_{j_1=1}^n\sum_{j_2=1}^n\cdots\sum_{j_r=1}^n$,
\[
T^{m_1,\dots,m_r}_{N_1,\dots,N_k}:=\sum_{\p{\mathbf{K}_1,\dots,\mathbf{K}_r}\in I^{\mb m}_{\mb N}}\sum_{j_1,\dots,j_r} \frac{\partial^r X_i}{\partial z_{j_1}\cdots\partial z_{j_r}}
\frac{\partial^{m_1}\varphi_{j_1}}{\partial \bm{z}_{\mathbf{K}_1}}\cdots
\frac{\partial^{m_r}\varphi_{j_r}}{\partial \bm{z}_{\mathbf{K}_r}}.
\]
\end{enumerate}
Then, the order-$k$ variational equation along $\bm\psi = \left\{\bm{\psi}\left(t\right)\right\}$ is summarised in the following:
\[
\frac{d}{dt}\frac{\partial^k \varphi_i}{\partial z_{N_1}\partial z_{N_2}\cdots \partial z_{N_k}} = \sum_{r=1}^k  \sum_{m_1, \dots , m_r} T^{m_1,\dots,m_r}_{N_1,\dots,N_k}, \quad
i,N_1,\dots,N_k\in \brr{1,\dots,n},
\]
indices in $ \sum_{m_1, \dots , m_r}$ constrained by $m_1 \le \dots \le m_r$ and $\sum_i^r m_i=k$. $\hfill\square$
\end{corollary}

\eqref{VEkredux} in Proposition \ref{nonlinearprove2} effectively settles the entries for lower $n$ rows in $A_{\mathrm{LVE}_{\psi}^k} $ and the first $n$ columns in $\Phi_k$. Let us now find the rest of the matrices.

\begin{proposition}[Explicit version of $\mathrm{LVE}_{\psi}^k$] \label{LVEprop} Still following Notation \ref{notaAY}, the infinite system
\begin{equation} \label{LVE}\tag{$\mathrm{LVE}_{\psi}$}
\fbox{$\dot X = A_{\mathrm{LVE}_{\psi}}X,$} \qquad\qquad A_{\mathrm{LVE}_{\psi}}:=A\odot \expodot \Id_n,
\end{equation}
has $\Phi:=\expodot Y$ as a solution matrix. Hence, for every $k\ge 1$,
\begin{enumerate}
\item the lower-triangular recursive $D_{n,k}\times D_{n,k}$ form for $\mathrm{LVE}_{\psi}^k$ is $\dot{Y} = A_{\mathrm{LVE}_{\psi}^k}Y$, its system matrix being obtained from the first $k$ row and column blocks of $A_{\mathrm{LVE}_{\psi}}$:
\begin{equation} \label{ALVE}  A_{\mathrm{LVE}_{\psi}^k} =  \p{
\begin{tabular}{ccc}
$\binom{k}{k-1}A_1\odot \Id_n^{\odot {k-1}}$ & \\
\cline{2-3} $\binom{k}{k-2}A_2\odot \Id_n^{\odot k-2}$  & \multicolumn{2}{|c|}{}\\
\vdots  & \multicolumn{2}{|c|}{$A_{\mathrm{LVE}_{\psi}^{k-1}}$} \\
$\binom{k}{0} A_k$  & \multicolumn{2}{|c|}{} \\\cline{2-3}
\end{tabular}
} ,\end{equation}
\item and the principal fundamental matrix for $\mathrm{LVE}_{\psi}^k$ is $\Phi_k$ from $\expodot Y$ in Notation \ref{notaAY}.
\end{enumerate}
\end{proposition}
\begin{proof}
\eqref{Aexp} in \ref{expfactor}, \eqref{VEphi} in Proposition \ref{nonlinearprove2},
and item (b) in Lemma \ref{expprops0} imply
\[ \dot{\overbracket[0.2pt]{\expodot Y}} = \dot Y \odot \expodot Y = \p{A\expodot Y} \odot \expodot Y = \p{A\odot \expodot \Id_n}\expodot Y . \]
The rest follows from Lemma \ref{AZproperties}.
\end{proof}

\begin{example}
For instance, for $k=5$ we have
\[
A_{\mathrm{LVE}_{\psi}^5}  = \p{
\begin{array}{ccccc}
5 A_1\odot \Id_n^{\odot 4} & & & & \\
10 A_2\odot \Id_n^{\odot 3}  & 4 A_1\odot \Id_n^{\odot 3}  & & & \\
10 A_3\odot \Id_n^{\odot 2}  & 6 A_2\odot \Id_n^{\odot 2} & 3 A_1\odot \Id_n^{\odot 2} & & \\
5 A_4 \odot \Id_n & 4 A_3 \odot \Id_n & 3 A_2 \odot \Id_n & 2 A_1 \odot \Id_n & \\
A_5 & A_4 & A_3 & A_2 & A_1
 \end{array}
 },
\]
and, using any of the equivalent  \eqref{expZrs}, \eqref{Zrkdef}, the principal fundamental matrix $\Phi_5$ is
{\small  \begin{equation} \label{PhiLVE5}
 \p{
\begin{array}{ccccc}
Y_1^{\odot 5} & & & & \\
10Y_1^{\odot 3}\odot Y_2 & Y_1^{\odot 4}  & & & \\
10Y_1^{\odot 2}\odot Y_3 + 15Y_1\odot Y_2^{\odot 2}& 6 Y_1^{\odot 2}\odot Y_2 & Y_1^{\odot 3} & & \\
10 Y_2\odot Y_3 + 5Y_1\odot Y_4 & 4 Y_1\odot Y_3 + 3 Y_2 \odot Y_2 & 3 Y_1\odot Y_2 & Y_1^{\odot 2} \\
Y_5 & Y_4 & Y_3 & Y_2 & Y_1
  \end{array}
 },
\end{equation}}hence \eqref{VEkredux} for $k=5$ is the lowest row in $A_{\mathrm{LVE}_{\psi}^5}$ times the leftmost column in $\Phi_5 $.
\end{example}

\subsection{Explicit solution and monodromy matrices for $\mathrm{LVE}^k_{\psi}$}

Let $T\subseteq \mathbb{P}^1_\nc$ be the domain for time variable $t$ in \eqref{DS} and $\gamma \subset T$ a closed path based at $t_0\in T$.
Assume $k=1$. If $Y_1$ is a fundamental matrix of first-order \eqref{VE},
analytic continuation along $\gamma$ yields
$ Y_1 \p{t_0} \xrightarrow[\mathrm{cont}]{\gamma}  Y_1\p{t_0} \cdot M_{1,\gamma}, $
$M_{1,\gamma}$ being the \emph{monodromy matrix} (\cite{Zoladek}) of \eqref{VE}. Assume $Y_1:=\Phi_1$ is the principal fundamental matrix for \eqref{VE}, any other solution matrix $\Psi_1$ recovered from $\Psi_1=Y_1\Psi_1\p{t_0}$. The non-linearised second-order equation, after Proposition \ref{nonlinearprove2}, is
\begin{equation}
\label{VE2} \tag{$\mathrm{VE}^2_{\psi}$}
\dot{Y_2} = A_1 Y_2 + A_2 \cdot \mathrm{Sym}^2 \p{Y_1}.
\end{equation}
Following Proposition \ref{LVEprop}, linearised completion $\mathrm{LVE}^2_{\psi}$ has principal fundamental matrix
\[ \Phi_2 = \p{\begin{array}{cc}  Y_1^{\odot 2} & \\ Y_2 & Y_1 \end{array}}. \]
A particular solution $
 Y_2 =   Y_1 \int Y_1^{-1} A_2  \mathrm{Sym}^2\p{Y_1}
$ of \eqref{VE2} is found via variation of constants,
which becomes a contour integral whenever time is taken along path $\gamma$:
\begin{equation}
\label{ve2mono}
Y_2 \xrightarrow[\mathrm{cont}]{\gamma} Q_{1,2,\gamma} :=  M_{1,\gamma} \int_{\gamma} Y_1^{-1} A_2  \mathrm{Sym}^2\p{Y_1} ,
\end{equation}
hence
\[
\Id_n =  \Phi_2\p{t_0} \xrightarrow[\mathrm{cont}]{\gamma}  \p{\begin{array}{cc} Y_1^{\odot 2}\p{t_0} M_{1,\gamma}^{\odot 2} & 0  \\  Y_1\p{t_0} Q_{1,2,\gamma}  & Y_1\p{t_0} M_{1,\gamma} \end{array}}
=
\Phi_2\p{t_0} \p{\begin{array}{cc} M_{1,\gamma}^{\odot 2}  &  0 \\ Q_{1,2,\gamma} & M_{1,\gamma} \end{array}},
\]
and $\qu{\gamma}\mapsto M_{i,\gamma}$ is a group morphism $\pi_1\p{T,t_0}\to \gl_{D_{n,i}}\p{\nc}$, hence for \emph{any} fundamental matrix \[
\Psi_2\p{t_0} \xrightarrow[\mathrm{cont}]{\gamma} \Psi_2\p{t_0} \p{\begin{array}{cc} M_{1,\gamma}^{\odot 2}  &  0 \\ Q_{1,2,\gamma} & M_{1,\gamma} \end{array}}; \]
therefore the monodromy of $\mathrm{LVE}^2_{\psi}$ along $\gamma$ will be
\begin{equation}
\label{lve2mono}
M_{2,\gamma} := \p{\begin{array}{cc} M_{1,\gamma} ^{\odot 2}  &  0 \\ Q_{1,2,\gamma} & M_{1,\gamma} \end{array}} = \p{\begin{array}{cc} M_{1,\gamma}^{\odot 2}  &  0 \\ M_{1,\gamma}\int_{\gamma} {Y_1^{-1} A_2Y_1^{\odot 2} }  & M_{1,\gamma} \end{array}}.
\end{equation}
Assume $k=3$. The principal fundamental matrix of $\mathrm{LVE}^3_{\psi}$ consists of the lower right $3\times 3$-block of \eqref{PhiLVE5}
and all solution matrices can be expressed
$\Psi_3 =\Phi_3 C$. Let us now find a solution to \begin{equation}
\label{VE3}
\dot{Y_3} = A_1 Y_3 + 3 A_2 Y_1\odot Y_2+A_3 \mathrm{Sym}^3 \p{Y_1},
\end{equation}
Same as before, variation of constants on \eqref{VE3}
yields another contour integral if $\tau\in \gamma$:
\begin{equation}
\label{ve3mono}
Y_3 \xrightarrow[\mathrm{cont}]{\gamma} Q_{1,3,\gamma} :=  M_{1,\gamma} \int_{\gamma} Y_1^{-1}  \p{3 A_2 Y_1\odot Y_2+A_3 \mathrm{Sym}^3 \p{Y_1} }d\tau .
\end{equation}
The remaining term of our monodromy matrix is a direct consequence of analytic continuation:
\[ 0 = 3 Y_1\p{t_0}\odot Y_2\p{t_0}  \xrightarrow[\mathrm{cont}]{\gamma}  3 M_{1,\gamma} \odot Q_{1,2,\gamma} = 3 M_{1,\gamma} \odot \p{ M_{1,\gamma}\int_{\gamma} Y_1^{-1}A_2Y_1^{\odot 2} }.\]
Our monodromy matrix is
\begin{equation}
\label{lve3mono}
M_{3,\gamma} \!:=\!
   \p{\!\begin{array}{ccc} M_{1,\gamma}^{\odot 3}  &  &   \\
3M_{1,\gamma} \odot  Q_{1,2,\gamma} & M_{1,\gamma}^{\odot 2} &  \\
 Q_{1,3,\gamma} & Q_{1,2,\gamma} & M_{1,\gamma} \end{array}}
= \p{\begin{tabular}{cc}
$M_{1,\gamma}^{\odot 3}$ &   \\ \cline{2-2}
 $3M_{1,\gamma} \odot  Q_{1,2,\gamma}$ & \multicolumn{1}{|c|}{\multirow{2}{*}{$M_{2,\gamma}$}} \\
$Q_{1,3,\gamma}$ &   \multicolumn{1}{|c|}{}  \\ \cline{2-2}
\end{tabular}}.
\end{equation}

The pattern is clear now. Assume we have computed solutions $Y_1,\dots,Y_{k-1}$ and performed continuation up to $k-1$:
{\small\[ \Phi_{k-1}  \xrightarrow[\mathrm{cont}]{\gamma} \Phi_{k-1} M_{k-1,\gamma}  := \Phi_{k-1} \p{ \begin{array}{ccccc}
 Q_{k-1,k-1,\gamma} & & &  & \\
 Q_{k-2,k-1,\gamma} &  Q_{k-2,k-2,\gamma} & &  &  \\
 \vdots &  \vdots & \ddots & &   \\
 Q_{2,k-1,\gamma} &  Q_{2,k-2,\gamma} & \cdots & Q_{2,2,\gamma} &  \\
 Q_{1,k-1,\gamma} &  Q_{1,k-2,\gamma} & \cdots & Q_{1,2,\gamma} & Q_{1,1,\gamma}
 \end{array} } , \] }where
\begin{equation} \label{theQs} Q_{r,s,\gamma} := \sum_{i_1+\dots+i_r=s} c^{s}_{i_1,\dots,i_r}Q_{1,i_1,\gamma}\odot Q_{1,i_2,\gamma} \odot \cdots \odot Q_{1,i_r,\gamma} , \qquad s\ge r \ge 2.
\end{equation}
Then, the fundamental matrix for $\p{\mathrm{LVE}^k_{\psi}}$ will be expressed in the form \eqref{PhiLVE}, its lower left block $Y_k$
being computable in terms of the blocks $Z_{2,k},\dots, Z_{k,k}$ above it (all of which involve $Y_1,\dots,Y_{k-1}$) in virtue of \eqref{VEkredux}: $Y_k=Y_1V_k$, which is continued into
$ Q_{1,k,\gamma} := M_{1,\gamma} \int_\gamma V_k,$ where $\dot{\overbracket[0.2pt]{V_k}}=
Y_1^{-1}\sum_{j=2}^k A_j Z_{j,k}$.
Upper terms $Z_{2,k},\dots, Z_{k,k}$ are continued into $Q_{2,k},\dots, Q_{k,k}$ as in \eqref{theQs}, $s$ replaced by $k$.
It is clear we have proven the following:
\begin{lemma}
The monodromy matrix $\Phi_k$ of $\mathrm{LVE}_\psi^k$ along closed path $\gamma$ is composed by the first $k$ row and column blocks in
\begin{equation}\label{Qs}
\expodot Q_\gamma := \expodot \p{\begin{tabular}{ccc|c}
$\cdots$ &   $0$ &  $0$ & $0$ \\
$\cdots$  & $Q_{1,2,\gamma}$ &  $Q_{1,1,\gamma}$ & $0$ \\\hline
$\cdots$  & $0$ &  $0$ & $0$
\end{tabular}}, \end{equation}
where $Q_{1,1,\gamma}:= M_{1,\gamma}$, blocks above the bottom row are computed according to \eqref{theQs} and
\begin{equation}\label{Qs2}
 Q_{1,s,\gamma} := M_{1,\gamma } \int_{\gamma} Y_1^{-1}\sum_{j=2}^s A_j Z_{j,s}, \quad 2\le s\le k . \qquad \square
\end{equation}
\end{lemma}
Hence it is clear the computation of a monodromy matrix follows a block order such as the one below, blocks in the bottom row requiring quadratures:
\begin{equation} \label{ordermat}
\begin{tabular}{ccccc}
$\ddots$ & & & & \\
\cline{2-2}$\cdots$ & \multicolumn{1}{|c|}{$7$} & & & \\
\cline{2-3}$\cdots$ & \multicolumn{1}{|c|}{$8$} & \multicolumn{1}{|c|}{$4$} & & \\
\cline{2-4}$\cdots$ & \multicolumn{1}{|c|}{$9$} & \multicolumn{1}{|c|}{$5$} & \multicolumn{1}{|c|}{$2$} & \\
\cline{2-5}$\cdots$ & \multicolumn{1}{|c|}{$10$} & \multicolumn{1}{|c|}{$6$} & \multicolumn{1}{|c|}{$3$} & \multicolumn{1}{|c|}{$1$}\\ \cline{2-5}
\end{tabular}
\end{equation}
Computing the monodromy matrix is concomitant to computing the fundamental matrix, i.e. said bottom-row quadratures must be both indefinite (yielding terms $Z_{1,s}$ to be used in the computation of $Z_{j,s}$ in \eqref{Qs2}) and contour integrals. See
\S \ref{samexample} for an example.

We assume there are two generators $\qu{\gamma}, \qu{\widetilde{\gamma}}\in \pi_1\p{T;t_0}$, yielding two different matrices:
\[ \gamma \longleftrightarrow Q_{\gamma}, \qquad \widetilde{\gamma} \longleftrightarrow Q_{\widetilde{\gamma}}.
\]
Commutativity of monodromy matrices now admits simple, compact formulation:
\begin{proposition}\label{commmon} Two monodromy matrices $M_{k,\gamma}$ and $M_{k,\tilde\gamma}$ for $\mathrm{LVE}_\psi^k$ commute if, and only if, their
previous blocks $M_{k-1,\gamma},M_{k-1,\tilde\gamma}$ commute and the additional properties hold
\[  \sum_{j=r}^k Q_{r,j,\gamma}Q_{j,k,\tilde\gamma} = \sum_{j=r}^k Q_{r,j,\tilde\gamma}Q_{j,k,\gamma}, \qquad \mbox{for every } r=1,\dots,k-1, \]
matrices defined as in \eqref{theQs} and \eqref{Qs2}. $\hfill\square$
\end{proposition}

\begin{remarks}
\begin{enumerate}
\item[\phantom{hola}]
\item The monodromy group of a linear system is contained in its
differential Galois group (e.g. \cite{SingerVanderput}). The motivation for the above Lemma and Proposition is to capitalise on this fact.
 This may in turn be a step towards future constructive incarnations of the Morales-Ramis-Sim\'o Theorem \ref{moralesramissimo}. The main obstacle  implementing Proposition \ref{commmon}, symbolico-computational issues aside,
 is the incertitude on whether $M_{k,\gamma}$ and $M_{k,\tilde\gamma}$ belong to the Zariski identity component
 $\mathrm{Gal}\p{\mathrm{LVE}^k_\psi}^\circ$; a sufficient condition for arbitrary order is fulfilment at order $1$,
$M_{1,\gamma},M_{1,\tilde\gamma}\in\mathrm{Gal}\p{\mathrm{VE}_\psi}^\circ$, itself an open problem in general.
\item All disquisitions and results on the variational jet in \cite{martsim1,martsim2} are referred to
the lower $n$-row strip for commutators of these monodromies. More specifically:
\begin{itemize}
\item what is called \emph{jet} therein is lower strip $Y$ in principal fundamental matrix $\Phi=\expodot Y$ for infinite system \eqref{LVE}, and we will use this terminology in the following Section;
\item morphism properties  imply monodromy matrices along path commutators equal matrix commutators:
\allowbreak $M_{k,\gamma_2^{-1}\gamma_1^{-1}\gamma_2\gamma_1} = M_{k,\gamma_2}^{-1}M_{k,\gamma_1}^{-1}M_{k,\gamma_2}M_{k,\gamma_1}$;
\item hence, ``jet commutation'' in \cite{martsim1,martsim2} amounts to lower strip
$Q_{k,\gamma_2^{-1}\gamma_1^{-1}\gamma_2\gamma_1}$ (that is $Y$ after passage along $\gamma_2^{-1}\gamma_1^{-1}\gamma_2\gamma_1$) equalling $\Id_n$.
\end{itemize}
Although \cite{martsim1,martsim2} clearly benefit from the use of automatic differentiation techniques (see also \cite{MakinoBerz}), it may be argued that
expressions such as those in \eqref{LVE} provide for a fuller control of the general structure of the whole variational complex
when it comes to symbolic computations, as well as a further check aid for the aforementioned techniques. See \S \ref{61} for an example. See also \cite{Simon2013} for a recent application to the Friedmann-Robertson-Walker Hamiltonian arising from Cosmology.
\end{enumerate}
\end{remarks}

\section{First integrals and higher-order variational equations}\label{firstintegrals}

Let $F:U\subseteq \nc^n\to \nc^n$ be a holomorphic function and $\bm{\psi}:I\subset \nc \to U$. Firstly, the flow $\varphi\p{t,\bm{z}}$ of $X$ admits, at least formally, Taylor expansion \eqref{taylorflow} along $\bm{\psi}$ which is expressible as
\begin{equation} \label{taylorphi}
\varphi\p{t,\bm{\psi}+\bm{\xi}} = \bm{\psi} + Y_1\bm{\xi}  + \frac12 Y_2 \bm{\xi}^{\odot 2} + \dots = \bm{\psi}+J_\psi \expodot\bm{\xi},
\end{equation}
where $J_\psi$ is the jet for flow $\varphi\p{t,\cdot}$ along $\bm\psi$, displayed as $Y$ in \eqref{Ydisplay} and defined in Notation \ref{notaAY} -- that is, the matrix
whose $\odot$-exponential $\Phi$ is a solution matrix for \eqref{LVE}.
Secondly, the Taylor series of $F$ along $\bm{\psi}$
can be written, cfr. \cite[Lemma 2]{ABSW} and Notation \ref{notalex},
\begin{equation} \label{taylorabsw} F\p{\bm{y}+\bm{\psi}}=F\left(\bm{\psi}\right)+\sum^{\infty}_{m=1} \frac{1}{m!}\left\langle F^{(m)}\left(\bm{\psi}\right)\,,\, \Sym^m \bm{y}\right\rangle.
\end{equation}
Basic scrutiny of Example \ref{examplesexp}(\textbf{3}), Lemma \ref{expformalseries} and \eqref{Fjet} trivially implies  \eqref{taylorabsw} can
be expressed as
$F\p{\bm{y}+\bm{\psi}} = M^\psi_F \expodot \bm{y}$, where
\[
M^\psi_F = J^\psi_F+F^{\p{0}}\!\p{\bm{\psi}} := \p{\begin{tabular}{ccc|c}
$\cdots$ &   $0$ &  $0$ & $0$ \\
$\cdots$  & $F^{\p{2}}\!\p{\bm{\psi}}$ &  $F^{\p{1}}\!\p{\bm{\psi}}$ & $F^{\p{0}}\!\p{\bm{\psi}}$ \\\hline
$\cdots$  & $0$ &  $0$ & $0$
\end{tabular}}\in\Mat^{1,n}\p{K},
\]
i.e. $J_F^\psi$ is the jet or horizontal strip of lex-sifted partial derivatives of $F$ at $\bm{\psi}$.

\begin{definition}
We call \begin{equation} \label{LVEstar}\tag{$\mathrm{LVE}_{\psi}^{\star}$}
\fbox{$\dot X = A_{\mathrm{LVE}^{\star}_{\psi}}X,$} \qquad\qquad A_{\mathrm{LVE}_{\psi}^\star}:=-\p{A\odot \expodot \Id_n}^T,
\end{equation}
the \textbf{adjoint} or \textbf{dual} variational system of \eqref{DS} along $\psi$. Same as in \eqref{LVE} and all throughout \ref{structure}, consideration of finite subsystems, namely the lowest $D_{n,k}\times D_{n,k}$ block, leads to specific notation $\p{\mathrm{LVE}_{\psi}^k}^{\star}$.
\end{definition}
The following is  immediate upon derivation of equation $\Phi_k\Phi_k^{-1}=\Id_{D_{n,k}}$:
\begin{lemma} $\p{\Phi^{-1}_k}^T$ is a principal fundamental matrix of $\p{\mathrm{LVE}_{\psi}^k}^\star$, $k\ge 1$.

Hence, $\lim_k\p{\Phi_k^{-1}}^T$, is a solution to \eqref{LVEstar}. $\hfill\square$
\end{lemma}

The following was proven in \cite{MoRaSi07a} and recounted in \cite[Lemma 7]{ABSW}, and may now be expressed in a simple, compact fashion:
\begin{lemma} \label{lemmanewproof}
Let $F$ and $\bm{\psi}$ be a holomorphic first integral and a non-constant solution of \eqref{DS} respectively. Let $V:=J_F^T$ be  the transposed jet of $F$ along $\bm{\psi}$. Then, $V$ is a solution of \eqref{LVEstar}.
\end{lemma}
\begin{proof}
Let us recall formal expansion \eqref{taylorphi}
and $F\p{\bm{y}}=J^{\psi}_F\expodot\bm{y}$ for every
$\bm{y}\in K^n$.
Let $\bm{\phi}=\varphi\p{t,\bm{\psi}+\bm{\xi}}$. We have, using Lemma \ref{lemaboutseries},
\[  F\p{ \bm\phi} =  F\p{\bm{\psi}+J_{\psi}\expodot\bm{\xi}}= M^\psi_F \expodot\p{J_{\psi}\expodot\bm{\xi}} = \p{M^\psi_F \expodot J_{\psi}}  \expodot \bm\xi,   \]
and $F\p{\bm\phi}$ is supposed to be constant, hence applying \eqref{LVE} and Lemma \ref{lemaboutseries}
\begin{eqnarray*} 0 &=&  \dot{\overbracket[0.2pt]{\p{M^\psi_F \expodot J_{\psi}}  }}\expodot \bm\xi =
 \p{\dot{\overbracket[0.2pt]{M^\psi_F}}+M^\psi_F A_{\mathrm{LVE}_{\psi}}}   \expodot J_{\psi}\expodot \bm\xi  \\
 &=& \p{\dot{\overbracket[0.2pt]{M^\psi_F}}+M^\psi_F A_{\mathrm{LVE}_{\psi}}}   \expodot \p{\bm\phi-\bm\psi
 },
 \end{eqnarray*}
 hence $\dot{\overbracket[0.2pt]{M^\psi_F}}+M^\psi_F A_{\mathrm{LVE}_{\psi}}=0$ leading us to the final result after transposing
 both sides.
\end{proof}
Compound the jet of field $X$, i.e. $A$ in Notation \ref{notaAY} and Proposition \ref{LVEprop}, with a ${}_{1,0}$ term $A_0$, equal to $X^{\p{0}}=X\p{\bm{\psi}}=\dot{\bm{\psi}}$:
\[
\widehat{A}:= \p{\begin{tabular}{ccc|c}
$\cdots$  &   $0$ &  $0$ &  $0$  \\
$\cdots$ &   $A_2$ &  $A_1$ & $A_0$ \\ \hline
$\cdots$ &  $0$ & $0$ &  $0$
\end{tabular}}, \qquad A_i:=X^{\p{i}}\!\p{\bm{\psi}}\in\Mat^{1,i}_n\p{K}.
\]
It is easy to check, via possibilities offered on $i_1$ and $j_1$ in \eqref{SymProdInf}, that
the symmetric product of $\widehat{A}$ with $\expodot \Id_n$ adds
only a relatively minor addendum to $A_{\mathrm{LVE}_\psi}$, namely a superdiagonal of
blocks $\binom{i}{i}A_0\odot \Id_n^{\odot i}\in\Mat_{n}^{i+1,i}$, $i\ge 1$, effectively rendering it block-Hessenberg:
\[
\widehat{A}_{\mathrm{LVE}_{\psi}}:=\widehat{A}\odot \expodot \Id_n = \lim_k \widehat{A}_{\mathrm{LVE}^k_{\psi}},
\]
where, isolating $A_{\mathrm{LVE}^k_{\psi}}$ within $\widehat{A}_{\mathrm{LVE}^k_{\psi}}$ by means of a solid line,
{\small \begin{eqnarray}\label{Atildek}
\widehat{A}_{\mathrm{LVE}^k_{\psi}} &:= &\p{ \begin{tabular}{cccccc}
$A_0\odot \Id^{\odot k}_n$ & & & & & \\
\Cline{1mm}{1-1}\multicolumn{1}{!{\vrule width 0.3mm}c}{$\binom{k}{k-1}A_1\odot \Id^{\odot k-1}_n$} &  & $\ddots$ &  & & \\
\multicolumn{1}{!{\vrule width 0.3mm}c}{$\vdots$} &  & \multicolumn{1}{c!{\vrule width 0.3mm}}{$\ddots$} & $A_0\odot\Id^{\odot 2}_n$ & & \\
\Cline{1mm}{4-4} \multicolumn{1}{!{\vrule width 0.3mm}c}{$\binom{k}{1}A_{k-1}\odot\Id_n$} &  & $\cdots$ & \multicolumn{1}{c!{\vrule width 0.3mm}}{$2A_1\odot\Id_n$} & $A_0\odot\Id_n$ & \\
\Cline{1mm}{5-5} \multicolumn{1}{!{\vrule width 0.3mm}c}{$A_k$} &  & $\cdots$ & $A_2$ & \multicolumn{1}{c!{\vrule width 0.3mm}}{$A_1$} & $A_0$ \\
\Cline{1mm}{1-5}
$0$ &  & $\cdots$ & $0$ &$0$ & $0$
\end{tabular}}\\
&=& \p{\begin{tabular}{|c|c}
\cline{1-1} $\binom{k}{k}X^{\p{0}}\p{\bm{\psi}} \odot  \Id_n^{\odot k}$ & \\
\cline{1-2} $ \binom{k}{k-1}X^{\p{1}}\p{\bm{\psi}} \odot \Id_n^{\odot k-1}$ & \multicolumn{1}{c|}{\multirow{4}{*}{$ \widehat{A}_{\mathrm{LVE}^{k-1}_{\psi}} $}} \\
\cline{1-1} $\binom{k}{k-2}X^{\p{2}}\p{\bm{\psi}} \odot  \Id_n^{\odot k-2}$ & \multicolumn{1}{c|}{} \\
\cline{1-1} $\vdots$ &  \multicolumn{1}{c|}{} \\
\cline{1-1} $\binom{k}{0} X^{\p{k}}\p{\bm{\psi}} \odot \Id_n^{\odot 0}$ &  \multicolumn{1}{c|}{} \\
\hline
\end{tabular} \nonumber
} .
\end{eqnarray}}Using the $M_k$--$\mathcal{M}_k$ notation in \cite{ABSW}, it is immediate to check that
\[ M_k^T=\Id_n^{\odot k-1}\odot \dot{\bm{\psi}}= \Id_n^{\odot k-1}\odot X\p{\bm{\psi}}, \qquad
 \widehat{A}_{\mathrm{LVE}^k_{\psi}} =\mathcal{M}^T_{k-1} \mbox{ for every }k\ge 1.
\]
A result in \cite{ABSW} using said notation is easier to prove in this setting. Indeed, the same reasoning underlying \eqref{Zaux2} applies to row $F^{\p{k}}$
and $ \frac{\partial}{\partial z_m} F^{\p{k}} = F^{\p{k+1}}\p{\bm{e}_m\odot \Id_n^{\odot k}}$;
following Lemma \ref{eLemma},
$$ \dot{\overbracket[0.2pt]{F^{\p{k}}}} =  F^{\p{k+1}}\sum_{m=1}^n \p{\bm{e}_m\odot \Id_n^{\odot k}} \dot{\psi_m} = F^{\p{k+1}}\p{ \dot{\bm{\psi}} \odot \Id_n^{\odot k}} = F^{\p{k+1}}\p{ A_0 \odot \Id_n^{\odot k}},  $$
implying {\small$\dot{\overbracket[0.2pt]{\p{F^{\p{k}}}^T}} =\p{ A_0 \odot \Id_n^{\odot k}}^T\p{F^{\p{k+1}}}^T$}; placing all terms on one side, and
observing Lemma \ref{lemmanewproof} and the transpose of expression \eqref{Atildek}, we obtain:
\begin{proposition}[{\cite[Th. 12]{ABSW}}] \label{ABSWlemma}
Let $F$, $\bm \psi$, $V$ as in Lemma \ref{lemmanewproof}. Then $\widehat{A}_{\mathrm{LVE}_{\psi}}^T V = 0$. \qed
\end{proposition}
Hence, blocks in $\bm V_1,\p{\bm V_2,\bm V_1}^T, \p{\bm V_3,\bm V_2,\bm V_1}^T, \dots$ having all entries in the base field $K$ and satisfying both equations in Proposition \ref{ABSWlemma} and \ref{lemmanewproof} are candidates for jet blocks $F^{\p{1}},F^{\p{2}},\dots$ of a formal first integral. These blocks belonging to the intersection of $\ker\widehat{A}_{\mathrm{LVE}^k_{\psi}}^T$ and the solution subspace
$\mathrm{Sol}_K \p{\mathrm{LVE}^k_{\psi}}^\star$ were called \emph{admissible} solutions of the order-$k$ adjoint system in \cite{ABSW}.

This takes us back to the end of Section \ref{powerseriessection}. Consider
\emph{gauge transformation} (\cite{Ap10a,ABSW,Au01a,Mo99a}) $\bm x= P \bm X$
transforming
linear system $\dot{\xi}=A_1{\xi}$ into equivalent
$$ \dot{{\Xi}}=P\qu{A_1} {\Xi} := \p{P^{-1}A_1P - P^{-1} \dot P}{\Xi} . $$
Using notation $Y_i=PX_i$, $J_\psi=PX$
and item (e) in Lemma \ref{AZproperties}, we recover the result already seen in previous references, summarised in the extension of gauge
transformations to higher dimensions via $P^{\odot k}$:
{\small\[ \expodot \p{X} =\expodot \p{P^{-1} J_\psi} = \expodot{P}^{-1}\expodot J_\psi = \mathrm{diag} \p{\cdots,\p{P^{-1}}^{\odot 2}, P^{-1}, 1} \expodot J_\psi, \]}and very simple application of properties seen so far extends the general structure of the gauge transformation to
$\Psi = \expodot P^{-1} \expodot J_\psi$:
{\small\begin{equation} \label{newLVE}
 \dot \Psi  = P\qu{A_{\mathrm{LVE}_\psi}} \Psi:= \p{\expodot P^{-1} A_{\mathrm{LVE}_\psi} \expodot P +\p{ \dot{\overbracket[0.5pt]{P^{-1}}}\odot \expodot P^{-1}}\expodot P} \Psi .
 \end{equation}}Second summand $\p{\dot{\overbracket[0.5pt]{P^{-1}}} \odot \expodot P^{-1}}\expodot P$ can be simplified into:
{\small\begin{equation} \label{Pdiag} \mathrm{diag}\p{\dots,k\qu{\dot{\overbracket[0.5pt]{P^{-1}}}\odot \p{P^{-1}}^{\odot k-1}}P^{\odot k},\dots,2\p{\dot{\overbracket[0.5pt]{P^{-1}}}\odot P^{-1}}P^{\odot 2},-P^{-1}\dot P,0} ,
\end{equation}}with $\dot{\overbracket[0.5pt]{P^{-1}}}=-P^{-1}\dot P P^{-1}$.
The above gauge transformation can be seen as the effect of transformation $\bm{z}=P\bm{Z}$ on the jet of \eqref{DS}. Given a first integral $F$ of the latter,
we may always assume $F\p{\bm\psi}=0$, which implies $M^{1,0}_F=0$ and,
as seen in \eqref{Fxnew} or in Lemma \ref{expprops0},
$$
F_P\p{\bm{Z}} =  J_F \p{\expodot{P}}\expodot  \bm Z .
$$
The jet of this formal series is
{\small\[ J_{F_P} =J_F \p{\expodot{P}}
= \p{\begin{tabular}{cccc|}
$\cdots$ & $F^{\p{0}}\!\p{\bm{\psi}}P^{\odot 3}$ & $F^{\p{2}}\!\p{\bm{\psi}}P^{\odot 2}$ &  $F^{\p{1}}\!\p{\bm{\psi}}P$  \\\hline
$\cdots$  & $0$ &  $0$ & $0$
\end{tabular}}\in\Mat^{1,n},
\]}and applying \eqref{newLVE}, Lemmae \ref{lemmanewproof} and \ref{ABSWlemma}, and identity
$ \p{P^{-1}}^{\odot k} \dot P^{\odot k} = -  \dot{\overbracket[0.5pt]{\p{P^{-1}}^{\odot k}}} P^{\odot k}, $
 we have just proven the following:
\begin{proposition}
The transposed jet $V_P:= J_{F_P}^T$ in the new variables must satisfy
\begin{equation} \label{finaleq}
\dot V_P = -P\qu{A_{\mathrm{LVE}_\psi}}^T V_P, \qquad \widehat{A}^T_{\mathrm{LVE}_\psi} \p{\expodot P^{-1}}^T V_P=0 . \qquad \square
\end{equation}
\end{proposition}
The key importance in practical examples resides in the choice of the particular solution $\psi$ and the reduction matrix $P$,
in order to render \eqref{finaleq} easier (or more convenient) to solve than its unreduced counterparts, Lemma \ref{lemmanewproof} and
Proposition \ref{ABSWlemma}: see also \cite{AW1,AW2}.

\section{Example}
\label{samexample}

The dynamics of the \emph{Swinging Atwood Machine (SAM)}, summarised in the diagram below,
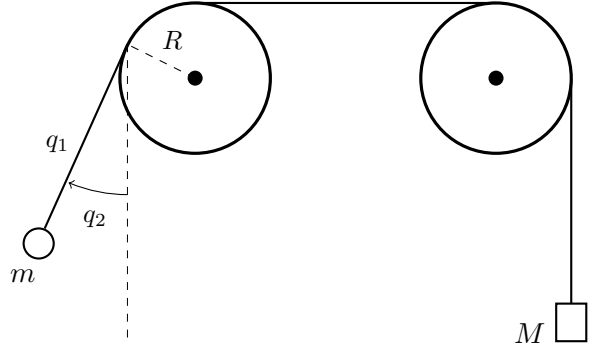
\[
\begin{tikzpicture}
\draw[very thick] (-2,0) arc (0:360:1);
\fill (-2.9,0) arc (0:360:0.1);
\fill (1.1,0) arc (0:360:0.1);
\draw[very thick] (2,0) arc (0:360:1);
\draw[thick] (-3,1) -- (1,1);
\draw[thick] (2,0) -- (2,-3);
\draw[thick] (1.8,-3.5) rectangle (2.2,-3);
\draw[dashed] (-3,0) -- (-3.9,0.45);
\draw[dashed,very thin] (-3.9,0.45) -- (-3.9,-3.5);
\draw [-to,shorten >=-1pt] (-3.9,-1.55) arc (270:248:2);
\draw[thick] (-3.9,0.45) -- (-5,-2);
\draw[thick] (-4.88,-2.2) arc (0:360:0.2);
\end{tikzpicture}
\setlength{\unitlength}{1cm}
\put (-5.65,3.9) {{\small{$R$}}}
\put (-6.7,1.6) {{\small{$q_2$}}}
\put (-7.2,2.6) {{\small{$q_1$}}}
\put (-7.8,0.8) {{ $m$}}
\put (-1.1,0.) {{ $M$}}
\]
are governed by
Hamiltonian
$$
\mathcal{H} = \frac12 \left[\frac{p_1^2}{M_t} + \frac{(p_2+Rp_1)^2}{mq_1^2}
\right] + gq_1(M-m\cos q_2)-gR(Mq_2-m\sin q_2),
$$
where $M_t=M+m+2I_p/R^2$ and $I_p$ is the pulley inertial momentum.
We know the following:
\begin{theorem}[\protect{\cite[Th. 7.5]{SAMPulleys}}]
\label{main}
For every physically consistent value of the parameters, regardless of $I_p$ and
$R$, ${\mathcal{H}}$ is meromorphically non-integrable.
\end{theorem}
Consider \emph{SAM}
\emph{without massive pulleys}, i.e. the limit case $I_p=R=0$ and $M_t=M+m$:
{\small\begin{equation} \label{hamsam}
\dot{\bm z}=X_{\mathcal{H}_w}\p{\bm z} := J\nabla \mathcal{H}_w\p{\bm{z}}, \qquad
\mathcal{H}_w = \frac12\left(\frac{p_1^2}{M+m}+\frac{p_2^2}{mq_1^2}\right) +
gq_1\left(M-m\cos q_2\right).
\end{equation}}
\begin{theorem}\label{previous} Define $\mu:= \frac Mm $ and $\mu_p:= \frac{p(p+1)}{p(p+1)-4}$, $p\in \nz$.
\begin{itemize}
\item[\textbf{1}.] \emph{(\cite[Th. 1]{Casasayas-1990})} If $M>m$ and $\mu \ne \mu_p$ for
every $p\geq 2$, then $X_{\mathcal{H}_w}$ is
non-integrable.
\item[\textbf{2}.] \emph{(\cite[(16)]{Tufillaro-1986})}
For $\mu=\mu_2=3$, \eqref{hamsam} is integrable with additional first integral:
\begin{equation} \label{otherintegral}
I = q_1^2\dot{q_2}\left(\dot{q_1}c-\frac{q_1\dot{q_2}}{2}s\right) + gq_1^2sc^2 =
 gq_1^2c^2s + p_2\frac{p_1q_1c - 2p_2s}{4m^2q_1}.
\end{equation}
\item[\textbf{3}.] \emph{(\cite[Theorem 4]{martsim1})} Degenerate cases
$\mu_p$, $p\geq 2$ in item \textbf{1} are non-integrable. $\hfill\square$
\end{itemize}
\end{theorem}
\noindent Canonical transformation
$\p{q_1,q_2,p_1,p_2} = \p{Q_1,
 \arccos Q_2,
 P_1,
 -P_2\sqrt {1 - Q_2^2}}
$
on $\mathcal{H}_w$ yields
\begin{equation}\label{samham2}
H=g Q_1 \p{M-m Q_2}+\frac{1}{2} \p{\frac{P_1^2}{M+m}-\frac{P_2^2 \left(Q_2^2-1\right)}{m Q_1^2}}.
\end{equation}
Let us apply our formulation to the variational systems for $H$. The non-meromorphic nature of the canonical transformation, along with other related issues, precludes us from extending the conclusions of \S \ref{61} to Hamiltonian \eqref{hamsam}. This will be the subject of further upcoming work.
\subsection{Monodromy matrices and integrability}\label{61}
Consider the particular solution $\bm \psi=\p{Q_1,Q_2,P_1,P_2}$ defined by
{\small\begin{equation} \label{solutionc}
\bm\psi\p{t} = \left(-\frac{g \p{t-1} t}{2} ,-1,-\frac{g \p{m+M} \p{2 t-1}}{2} ,\frac{g^2 m \p{t-1} t \left(C_1-2 C_1 t+t^2\right)}{4 \p{C_1-t}}\right).
\end{equation}}In the forthcoming calculations, any value of $C_1$ different from $0$ or $1$ will ensure the presence of logarithms in the fundamental matrix, and any value of $C_1$ different from $1/2$ will avoid division by zero. Choose, for instance, $C_1=1/3$.
Gauge transformation $\bm x = P\bm X$,
\[
P:=\frac{1}{\sqrt{M+m}}\p{
\begin{array}{cccc}
 -1 & 0 & 0 & 0 \\
 0 & \frac{2 \p{M+m} \p{1-3 t}^2}{9 g m \p{t-1}^2 t^2} & 0 & 0 \\
 0 & \frac{\p{M+m} \p{1+t} \p{3 t-1}}{9 \p{t-1}^2 t^2} & -M-m & 0 \\
 -\frac{g m \p{1+t}}{2  (3 t-1)} & -\frac{g F\p{t}}{90 (1-3 t)^2 (t-1) t} & 0 & \frac{9 g m (t-1)^2 t^2}{2\p{1-3 t}^2}
\end{array}
},
\]
and {\small$F=15 m (t-1) t (3 t-1)^3+M (1+t (16+15 t (-9+t (37+27 (-2+t) t))))$} transforms \eqref{VE} into
\begin{equation}\label{gaugeP}
\dot{\bm X}=P\qu{A_1}\bm X:=
\p{
\begin{array}{cccc}
 0 & -\frac{1}{9} \left(\frac{4}{\p{t-1}^2}-\frac{1}{t^2}\right) & 1 & 0 \\
 0 & 0 & 0 & 0 \\
 0 & 0 & 0 & 0 \\
 0 & \frac{8 M (1-3 t)}{405 m (t-1)^4 t^4} & \frac{1}{9} \left(\frac{4}{\p{t-1}^2}-\frac{1}{t^2}\right) & 0
\end{array}
}\bm X.
\end{equation}
Defining
\begin{eqnarray*}
G\p{t} &=&  \frac{2}{3} \p{2 t-1} \qu{1+5 \p{t-1} t \p{6 \p{t-1} t-1}} \\
&+& t \qu{3+5 t (3-2 t (11+3 t (2 t-5)))}-20 \p{t-1}^3 t^3 \ln\frac{t-1}{t},
\end{eqnarray*}
the principal fundamental matrix for \eqref{gaugeP} is $\Psi\p{t}:=\tilde{\Psi}\p{t}\cdot \p{\tilde{\Psi}\p{1/2}}^{-1}$, where
\[
\tilde{\Psi} :=e^{\int P\qu{A_1}}=
\p{
\begin{array}{cccc}
 1 & \frac{1+3 t}{9 \p{t-1} t} & t & 0 \\
 0 & 1 & 0 & 0 \\
 0 & 0 & 1 & 0 \\
 0 & -\frac{4 M G\p{t}}{405 m \p{t-1}^3 t^3} & -\frac{1+3 t}{9 (t-1) t} & 1
\end{array}
},
\]because $P\qu{A_1}$ commutes with $\int P\qu{A_1}$ --
an interesting topic for further study is the possible relationship between said commutation and the level of \emph{reduction} of the matrix
in the sense of \cite{AW1}, but this is not central to our present work. Our intention is to perform analytical continuation along paths homotopic to $\gamma_1,\gamma_2$ shown below:
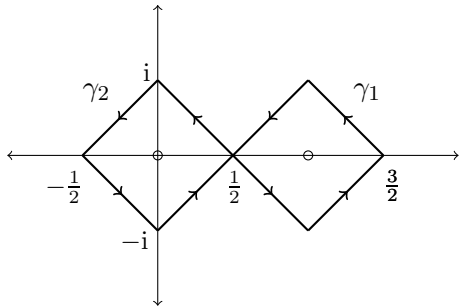
\[
\begin{tikzpicture}
\draw[<->](-2,0) -- (4,0);
\draw[<->](0,-2) -- (0,2);
\draw [-to,shorten >=-1pt,thick] (1,0) -- (1.5,-.5);
\draw [thick] (1.5,-0.5) -- (2,-1);
\draw [-to,shorten >=-1pt,thick] (2,-1) -- (2.5,-0.5);
\draw [thick] (2.5,-0.5) -- (3,0);
\draw [-to,shorten >=-1pt,thick] (3,0) -- (2.5,.5);
\draw [thick] (2.5,0.5) -- (2,1);
\draw [-to,shorten >=-1pt,thick] (2,1) -- (1.5,0.5);
\draw [thick] (1.5,0.5) -- (1,0);
\draw [-to,shorten >=-1pt,thick] (1,0) -- (.5,.5);
\draw [thick] (.5,.5) -- (0,1);
\draw [-to,shorten >=-1pt,thick] (0,1) -- (-.5,0.5);
\draw [thick] (-.5,0.5) -- (-1,0);
\draw [-to,shorten >=-1pt,thick] (-1,0) -- (-.5,-.5);
\draw [thick] (-.5,-.5) -- (0,-1);
\draw [-to,shorten >=-1pt,thick] (0,-1) -- (.5,-0.5);
\draw [thick] (.5,-0.5) -- (1,0);
\draw (2.06,0) arc (0:360:0.06);
\draw (0.06,0) arc (0:360:0.06);
\end{tikzpicture}
\setlength{\unitlength}{1cm}
\put (-5.5,1.5) {{\small{$-\frac 12$}}}
\put (-3.1,1.5) {{\small{$\frac 12$}}}
\put (-1,1.52) {{\small{$\frac 32$}}}
\put (-1,1.52) {{\small{$\frac 32$}}}
\put (-4.2,3) {{\small{$\ri$}}}
\put (-4.5,0.77) {{\small{$-\ri$}}}
\put (-1.4,2.8) {{$\gamma_1$}}
\put (-5,2.8) {{$\gamma_2$}}
\]
Along $\gamma_i$, all terms outside of the diagonal in $\Psi$ vanish except for term $\p{M_{1,\gamma_i}}_{4,2}$ which can be obtained from
\eqref{gaugeP} by assuming $X_2=1,X_3=0$, which yields $\frac{8 M (3t-1)}{405 m \p{t-1}^4 t^4}+\dot{X_4}=0$
and thus
{\small
\begin{equation}\label{monvalues}
 \p{M_{1,\gamma_1}}_{4,2} = - \int_{\gamma_1} \frac{8 M (3t-1)}{405 m \p{t-1}^4 t^4} = \frac{32 \ri M \pi }{81 m}, \quad \p{M_{1,\gamma_2}}_{4,2} =  -\frac{32 \ri M \pi }{81 m}.
 \end{equation}
 }Order-$2,3$ monodromies are given by \eqref{ve2mono}, \eqref{ve3mono} respectively. The only change therein is replacing
$A_1$, $A_2$, $A_3$ by their gauge transforms given by \eqref{newLVE}, \eqref{Pdiag} i.e. the lower row block in $P\qu{A_{\mathrm{LVE}^3_\psi}}$, equal to
{\small\[
\p{\begin{array}{ccc}
 3 \qu{Q^{\odot 3} \p{A_1\odot \Id_n^{\odot 2}}+\dot Q\odot Q^{\odot 2}}P^{\odot 3} & & \\
3 Q^{\odot 2} \p{A_2 \odot \Id_n} P^{\odot 3} & 2 \qu{Q^{\odot 2} \p{A_1 \odot \Id_n}  +\dot Q\odot Q}P^{\odot 2} & \\
P^{-1}A_3P^{\odot 3} & P^{-1}A_2P^{\odot 2} & P\qu{A_1}
\end{array}},
\]
}where $Q:=P^{-1}$. Let $Y_1=\Psi$. Following order \eqref{ordermat}, we compute \eqref{ve2mono}, \eqref{ve3mono} and $M_{3,\gamma_1},M_{3,\gamma_2}$. Let us check whether our third-order monodromy matrices commute. We have
$C:=M_{3,\gamma_1}M_{3,\gamma_2}-M_{3,\gamma_2}M_{3,\gamma_1}$ equal to zero except for the following terms:
\begin{equation} \label{m3values}
 C_{34,5}=-\frac{\left(\frac{25600}{85293}-\frac{8960 }{85293}\ri\right) M^2 \pi ^2}{g^2 m^2 (M+m)}, \qquad
C_{34,11} =\frac{6815744  M^3 \pi ^3\ri}{2187 g^2 m^3 (M+m)}.
\end{equation}
This, coupled with the fact that $M_{k,\gamma_i}\in\mathrm{Gal}\p{\mathrm{LVE}_\psi^k}^\circ$ for every $k$ (since
$M_{1,\gamma_i}$, being unipotent, belong to $\mathrm{Gal}\p{\mathrm{VE}_\psi}^\circ$ and fundamental matrices $\Psi_k$ are
obtained from quadratures), allows us to complement the non-integrability proof
in \cite{martsim1} using linearised variational equations on \eqref{samham2}.
\begin{remarks}
\begin{enumerate}
\item[\phantom{ }]
\item[\textbf{1.}]  Same reasoning can be applied to pre-gauge monodromies, although calculations are more cumbersome; defining $P_k:=\mathrm{diag}\p{P^{\odot k},\dots,P}$, we have
\[ \Psi_k \xrightarrow[\mathrm{cont}]{\gamma_i} \Psi_k M_{k,\gamma_i} \Rightarrow \tilde{\Psi}_k:= P_k\Psi_k\p{P_k\Psi_k}_{t=1/2}^{-1} \xrightarrow[\mathrm{cont}]{\gamma_i} P_k\Psi_k M_{k,\gamma_i} \p{P_k\Psi_k}_{t=1/2}^{-1} ,\]
thereby rendering monodromy $\tilde{M}_{k,\gamma_i}=\p{P_k\Psi_k}_{t=1/2}M_{k,\gamma_i} \p{P_k\Psi_k}_{t=1/2}^{-1} $. The reader may check that $\tilde{M}_{1,\gamma_i}$ have the same structure as $M_{1,\gamma_i}$, albeit with $\pm\frac{\ri g^2 m M \pi }{2 (M+m)}$ in lieu of \eqref{monvalues}, and  the non-zero terms in $\tilde{C}:=\tilde{M}_{3,\gamma_1}\tilde{M}_{3,\gamma_2}-\tilde{M}_{3,\gamma_2}\tilde{M}_{3,\gamma_1}$ share the same indices as \eqref{m3values}:
\begin{equation} \label{m3values2}
\tilde{C}_{34,5} = \frac{\left(\frac{50}{117}-\frac{35 \ri}{234}\right) g m M^2 \pi ^2}{(M+m)^2}, \qquad \tilde{C}_{34,11}=\frac{4992  g^2 m M^3 \pi^3\ri}{(m+M)^3}.
\end{equation}
\item[\textbf{2.}] Although we used monodromies, we can also use the presence of $\ln\p{3t-1}$ in $Y_3$ and
Picard-Vessiot extension {\small$\nc\p{t}\p{\ln t, \ln\p{t-1},\ln\p{3t-1}}\mid \nc\p{t}\p{\ln t, \ln\p{t-1}}$} to glean the structure of a generic Galois group matrix on the fundamental matrix $\Psi_3$ or its pre-gauge counterpart $P_3\Psi_3$; see
e.g. \cite{MSS}.
\end{enumerate}
\end{remarks}
We therefore have the following result:
\begin{theorem} $H$ in \eqref{samham2} is not meromorphically integrable for any $M,m>0$.$\hfill\square$
\end{theorem}

\subsection{Formal first integrals and admissible solutions}
Let us now try to apply gauge transforms to the adjoint system for the same Hamiltonian. We have a particular solution
$
\bm \psi=\left(-\frac{g}{2}  \p{t-1} t,-1,-\frac{g}{2}  \p{M+m} \p{2 t-1},-\frac{g^2m}{4}   \p{t-1} t^2\right),
$
which corresponds to special case $C_1=0$ in \eqref{solutionc}, and gauge transformation $\bm x = P\bm X$ with
{\small\[
P:=\p{
\begin{array}{cccc}
\frac{1-t}{\sqrt{M+m}}& 0 & 0 & 0 \\
0 & \frac{\sqrt{M+m}}{g m \p{t-1}^2} & 0 & 0 \\
-\sqrt{M+m} & \frac{\sqrt{M+m}}{\p{t-1}^2} & \frac{\sqrt{M+m}}{1-t} & 0 \\
\frac{g m \p{1-t}}{\sqrt{M+m}} & \frac{g \p{3 m-M}}{12 \p{t-1}\sqrt{M+m} } -\frac{g t}{4} \sqrt{M+m}  & 0 & \frac{g m \p{t-1}^2}{\sqrt{M+m}}
\end{array}
},
\]}transforms \eqref{VE} into the parameter-free, simplified system $\dot{\bm X}=P\qu{A_1}\bm X$ where
\[
P\qu{A_1} =
\p{
\begin{array}{cccc}
0 & -\frac{1}{\p{t-1}^3} & \frac{1}{\p{t-1}^2} & 0 \\
0 & 0 & 0 & 0 \\
0 & 0 & 0 & 0 \\
0 & 0 & \frac{1}{\p{t-1}^3} & 0
\end{array}
}.
\]
The principal fundamental matrix for this system is
\[
\Psi :=
\p{
\begin{array}{cccc}
1 & \frac{1}{2}\p{\frac{1}{\p{t-1}^2}-1} & \frac{t}{\p{1-t}} & 0 \\
0 & 1 & 0 & 0 \\
0 & 0 & 1 & 0 \\
0 & 0 & \frac{\p{t-2}t}{2\p{t-1}^2} & 1
\end{array}
},
\]
and $P\qu{A_{\mathrm{LVE}^3_\psi}}$ is computed as in \S \ref{61}.
In the following table, the first two columns
correspond to order $k$ and total solution space dimension $D_{4,k}$. The latter two
display the dimension of the subspace of rational solutions (i.e. those having their entries in the base field $K=\nc\p{t}$), and the
subspace of those among the former satisfying \eqref{finaleq}, respectively:
{\small\[
\begin{tabular}{c||c|cc}
  \toprule[1.5pt]
  {$k$} &  {$\dim_\nc \mathrm{Sol}P\qu{\mathrm{LVE}_\psi^k}^\star$}&  {$\dim_\nc \mathrm{Sol}_{K}P\qu{\mathrm{LVE}_\psi^k}^\star$} & {$\dim_\nc \mathrm{Sol}_{\mathrm{adm}}\p{P\qu{\mathrm{LVE}_\psi^k}^\star}$} \\
  \midrule
  $1$ &  $4$ &  $4$  & $\mathbf{3}$\\
  $2$ &  $14$ &  $14$  & $\mathbf{9}$ \\
  $3$ &  $34$ &  $32$  & $\mathbf{17}$ \\
  \bottomrule[1.5pt]
\end{tabular}
\]}Using solution \eqref{solutionc} in \S \ref{61} and  the other $P$ given in \eqref{gaugeP}, however, \allowbreak$\dim_\nc \mathrm{Sol}_{K}\break P\qu{\mathrm{LVE}_\psi^k}^\star$ is considerably reduced and we obtain the following table, displaying lower bounds on the amount of admissible solutions:
{\small\[
\begin{tabular}{c||c|cc}
  \toprule[1.5pt]
  {$k$} &  {$\dim_\nc \mathrm{Sol}P\qu{\mathrm{LVE}_\psi^k}^\star$}&  {$\dim_\nc \mathrm{Sol}_{K}P\qu{\mathrm{LVE}_\psi^k}^\star$} & {$\dim_\nc
 \mathrm{Sol}_{\mathrm{adm}}\p{P\qu{\mathrm{LVE}_\psi^k}^\star}$} \\
  \midrule
  $1$ &  $4$ &  $3$  & $\mathbf{2}$\\
  $2$ &  $14$ &  $9$  & $\mathbf{5}$ \\
  $3$ &  $34$ &  $19$  & $\mathbf{9}$ \\
  \bottomrule[1.5pt]
\end{tabular}
\]}

\section*{Acknowledgments}
The author is supported by Grants \textsc{MTM2006-05849/\break Consolider} and \textsc{MTM2010-16425}. The author is also indebted to
A. Aparicio-Monforte, J. J. Morales-Ruiz, J.-P. Ramis, C. Sim\'o and J.-A. Weil for useful suggestions and discussions.

\renewcommand{\refname}{REFERENCES}

\medskip
Received  March  2013; revised April 2014.
\medskip
\end{document}